\documentclass[journal]{IEEEtran}
\usepackage{bbm}
\usepackage{array}
\usepackage{dcolumn}
\usepackage{epsfig}
\usepackage[intlimits]{amsmath}
\usepackage{amsmath, amsfonts, yhmath, bm}
\usepackage{amssymb}
\usepackage{psfrag}
\usepackage{color,soul}
\usepackage{xcolor,cancel}
\usepackage[normalem]{ulem}
\usepackage{enumerate}
\usepackage{stackengine}
\usepackage[noadjust]{cite}
\usepackage{graphicx}
\usepackage{caption}
\usepackage{subcaption}
\usepackage{cite}
\usepackage[font=footnotesize]{caption}
\usepackage{etoolbox}
\newtheorem{thm}{Theorem}

\newtheorem{proof}{Proof}

\definecolor{purple}{rgb}{0.6,0.0,0.7}

\definecolor{cyan}{rgb}{0.2,0.6,0.7}

\newcommand {\myvec}[1] {{\mbox{\boldmath $#1$}}}
\newcommand {\mymat}[1]  {{\mbox{\boldmath $#1$}}}

\newcommand {\mS} {\mymat{S}}
\newcommand {\mX} {\mymat{X}}
\newcommand {\A} {\mymat{A}}
\newcommand {\Omeg} {\mymat{\Omega}}
\newcommand {\U} {\mymat{U}}
\newcommand {\hA} {\widehat{\A}}
\newcommand {\bA} {\mybar{\A}}
\newcommand {\bB} {\mybar{\B}}
\newcommand {\bP} {\mybar{\P}}
\newcommand {\bS} {\mybar{\mS}}
\newcommand {\bX} {\mybar{\X}}
\newcommand {\bC} {\mybar{\C}}
\newcommand {\B} {\mymat{B}}
\newcommand {\mBeta} {\mymat{\Psi}}

\newcommand {\brvB} {\overline{\overline{{\B}}}}
\newcommand {\hB} {\widehat{\B}}
\newcommand {\tB} {\widetilde{\B}}

\newcommand {\C} {\mymat{C}}
\newcommand {\Chi} {\mymat{\mathcal{X}}}
\newcommand {\D} {\mymat{D}}
\newcommand {\tD} {\widetilde{\D}}

\newcommand {\E} {\mymat{E}}
\newcommand {\tE} {\widetilde{\E}}

\newcommand {\G} {\mymat{G}}
\renewcommand {\H} {\mymat{H}}
\newcommand {\tH} {\widetilde{\H}}

\renewcommand {\P} {\mymat{P}}
\newcommand {\PI} {\mymat{\Pi}}
\newcommand {\Q} {\mymat{Q}}
\newcommand {\brvQ} {\mybar{\mybar{{\Q}}}}

\newcommand {\Ga} {\mymat{\Gamma}}
\newcommand {\I} {\mymat{I}}
\newcommand {\Ka} {\mymat{\mathcal{K}}}
\newcommand {\X} {\mymat{X}}

\newcommand {\ue} {\myvec{e}}
\newcommand {\ua} {\myvec{a}}
\newcommand {\ub} {\myvec{b}}
\newcommand {\tb} {\widetilde{\ub}}
\newcommand {\tbeta} {\widetilde{\myvec{\beta}}}
\newcommand {\ubeta} {\myvec{\psi}}
\newcommand {\uc} {\myvec{c}}

\newcommand {\us} {\myvec{s}}
\newcommand {\bs} {\bar{\us}}

\newcommand {\uphi} {\myvec{\phi}}

\newcommand {\Rset} {\mathbb{R}}
\newcommand {\Cset} {\mathbb{C}}
\newcommand {\Tr} {\text{Tr}}
\newcommand {\tps} {\text{T}}

\makeatletter
\newsavebox\myboxA
\newsavebox\myboxB
\newlength\mylenA

\newcommand*\mybar[2][0.75]{%
    \sbox{\myboxA}{$\m@th#2$}%
    \setbox\myboxB\null
    \ht\myboxB=\ht\myboxA%
    \dp\myboxB=\dp\myboxA%
    \wd\myboxB=#1\wd\myboxA
    \sbox\myboxB{$\m@th\overline{\copy\myboxB}$}
    \setlength\mylenA{\the\wd\myboxA}
    \addtolength\mylenA{-\the\wd\myboxB}%
    \ifdim\wd\myboxB<\wd\myboxA%
       \rlap{\hskip 0.5\mylenA\usebox\myboxB}{\usebox\myboxA}%
    \else
        \hskip -0.5\mylenA\rlap{\usebox\myboxA}{\hskip 0.5\mylenA\usebox\myboxB}%
    \fi}
\makeatother

\def\comment#1{}

\DeclareSymbolFont{grb}{OML}{cmm}{b}{it}
\DeclareMathSymbol{\lambdab}{\mathord}{grb}{"15}

\newcommand{\stkout}[1]{
\color{red}\ifmmode\text{\sout{\ensuremath{#1}}}\else\sout{#1}\fi\color{black}}

\newcommand{\addra}{}
\newcommand{\delra}{\comment}
\newcommand{\mathhlra }{}

\newcommand{\addrb}{}

\begin{document}

\title{The Extended ``Sequentially Drilled" Joint Congruence Transformation and its Application in Gaussian Independent Vector Analysis}

\author{Amir Weiss$^1$, Arie Yeredor$^1$, Sher Ali Cheema$^2$, and Martin Haardt$^2$

\thanks{$^1$ Authors with Tel-Aviv University, School of Electrical Engineering,
P.~O.~Box 39040, Tel-Aviv 69978, Israel, e-mail:
amirwei2@mail.tau.ac.il, arie@eng.tau.ac.il }
\thanks{
$^2$ Authors with
Ilmenau University of Technology,
Communications Research Laboratory,
P.~O.~Box 10~05~65, D-98684 Ilmenau, Germany,
e-mail: \{sher-ali.cheema, martin.haardt\}@tu-ilmenau.de,
phone: +49 (3677) 69-2613,
fax: +49 (3677) 69-1195,
WWW: http://www.tu-ilmenau.de/crl.}
}

\maketitle

\begin{abstract}
Independent Vector Analysis (IVA) has emerged in recent years as an extension of Independent Component Analysis (ICA) into multiple sets of mixtures, where the source signals in
each set are independent, but may depend on source signals in the other sets. In a semi-blind IVA (or ICA) framework, information regarding the probability distributions of the
sources may be available, giving rise to Maximum Likelihood (ML) separation. In recent work we have shown that under the multivariate Gaussian model, with arbitrary temporal
covariance matrices (stationary or non-stationary) of the source signals, ML separation requires the solution of a ``Sequentially Drilled" Joint Congruence (SeDJoCo)
transformation of a set of matrices, which is reminiscent of (but different from) classical joint diagonalization. In this paper we extend our results to the IVA problem,
showing how the ML solution for the Gaussian model (with arbitrary covariance and cross-covariance matrices) takes the form of an extended SeDJoCo problem. We formulate the
extended problem, derive a condition for the existence of a solution, and propose two iterative solution algorithms. In addition, we derive the induced Cram\'{e}r-Rao Lower
Bound (iCRLB) on the resulting Interference-to-Source Ratios (ISR) matrices, and demonstrate by simulation how the ML separation obtained by solving the extended SeDJoCo problem
indeed attains the iCRLB (asymptotically), as opposed to other separation approaches, which cannot exploit prior knowledge regarding the sources’ distributions.
\end{abstract}


\section{Introduction}
\label{sec_Intro} Joint matrix transformations and tensor decompositions are important fundamental algebraic tools in a broad range of signal processing fields, such as Blind
Source Separation (BSS, \cite{yeredor2002non,yeredor2005using,cardoso1990eigen,cichocki2009nonnegative}), Independent Vector Analysis (IVA,
\cite{li2011joint,anderson2014independent,via2011maximum}), data mining (\cite{sun2006window,acar2006collective,kolda2008scalable,morup2011applications}) and Multi-User
Multiple-Input Multiple-Output (MU-MIMO) systems in wireless communications (\cite{stankovic2008generalized,song2010using,de2006tensor,favier2012tensor}). In many data analysis
problems, special internal structures can be revealed by applying some transformations or decompositions to sets of matrices (or to tensors) constructed from the available data.

One common example is the use of Approximate Joint Diagonalization (AJD), which is instrumental in the context of BSS and independent component analysis (ICA). In many separation
algorithms (e.g., JADE \cite{cardoso1989source}, SOBI \cite{belouchrani1997blind}, ACMA \cite{van1996analytical}), the demixing-matrix is estimated as the matrix which best
``jointly diagonalizes", by some chosen criterion, a set of $N$ ``target-matrices" $\Q_1,\ldots,\Q_N$, constructed (in some way) from the observed mixtures (e.g., a set of sample
correlation matrices  at different lags). The procedure of AJD of such a set of $N$ matrices, each of dimensions $K \times K$, can be viewed as a symmetric canonical decomposition
(e.g., \cite{de2006link}), representing (or approximating) the respective three-way $K \times K \times N$ tensor as the sum of $K$ rank-1 tensors.

\subsection{From SeDJoCo to Extended SeDJoCo}
A particular case of this paradigm is the ``Sequentially Drilled" Joint Congruence (SeDJoCo) transformation, also termed ``Hybrid Exact-Approximate joint Diagonalization (HEAD)"
in (\cite{yeredor2009hybrid,yeredor2010blind}). SeDJoCo essentially prescribes the likelihood equations for Maximum Likelihood (ML) estimation in the semi-blind separation
scenario under a multivariate Gaussian model. More specifically, consider the classic linear, static, memoryless mixture model
\begin{equation}
	\X = \A\mS,
\end{equation}
where $\A \in \Rset^{K\times K}$ is the unknown, deterministic (invertible) mixing-matrix, $\mS = \left[\us_1\;\cdots\;\us_K\right]^{\tps}\in\Rset^{K\times T}$ is the sources
matrix of $K$ statistically independent source signals ($\us_1,\ldots,\us_K\in\Rset^{T}$) each of length $T$ (where $(\cdot)^{\tps}$ denotes the transpose), and
$\mX\in\Rset^{K\times T}$ is the observation matrix which contains the observed mixture signals. It has been shown (e.g., \cite{pham2001blind,yeredor2010blind} and
\cite{comon2010handbook} (chapter 7)), that when the source signals are zero-mean Gaussian, each with a known temporal covariance matrix $\C_k\triangleq E\left[\us_k
\us_k^{\tps}\right]\in\Rset^{T\times T}$ (all distinct from each other), the ML estimate $\hB$ of $\B \triangleq \A^{-1}$ can be obtained (up to an inevitable sign ambiguity) as
follows:
\begin{enumerate}[I.]
	\item Construct $N=K$\footnote{Note that unlike the general (and heuristic) AJD approach, in SeDJoCo the number of target-matrices equals the number of sources} symmetric
	target-matrices as
	\begin{equation}
		\Q_k \triangleq \frac{1}{T}\X\C_{k}^{-1}\X^{\tps}\in\Rset^{K\times K},\;\;\; \forall k\in\{1,\ldots,K\}.
	\end{equation}
	\item Find a matrix $\hB\in\Rset^{K\times K}$  that satisfies the following set of $K^2$ equations ($K$ vector equations, $K$ elements each)
	\begin{equation} \label{sedjoco_condition}
		\hB\Q_k\hB^{\tps}\ue_k=\ue_k,\;\;\; \forall k\in\{1,\ldots,K\},
	\end{equation}
	where the pinning vector $\ue_k$ denotes the $k$-th column of the $K\times K$ identity matrix.
\end{enumerate}
Condition \eqref{sedjoco_condition} is equivalent to requiring that the matrices $\D_k\triangleq \hB\Q_k\hB^{\tps}$ each satisfy $\D_k\ue_k=\ue_k$ (for $k=1,\ldots,K$), meaning
that the vector $\ue_k$ is an eigenvector of the matrix $\D_k$ with eigenvalue $1$. We term this structure a ``drilled" structure, hence the name of this transformation. Notice
that $\ue_k^{\tps}\D_k=\ue_k^{\tps}$ also holds, by the symmetry of all $\D_k$. This special ``drilled" structure of SeDJoCo, and its interpretation as a tensor decomposition,
are illustrated in Fig. \ref{fig1tensor} for the case of $K=3$. Here $\D_1$, $\D_2$, and $\D_3$ are treated as the first, second, and third frontal slices of the tensor
$\bm{\mathcal{D}}$, respectively. The target matrices $\Q_1$, $\Q_2$, and $\Q_3$ are associated with $\bm{\mathcal{Q}}$ in the same manner.

Another context in which the SeDJoCo solution is useful is the MU-MIMO broadcast channel Coordinated Beamforming (CBF) \cite{yeredor2012sequentially,song2010using}, when a
transmitter with $K$ antennas transmits data to $N\le K$ users, each with $K$ antennas. In order to find the beamformer which perfectly eliminates all inter-users interference, a
very similar (complex-valued) formulation of the SeDJoCo equations is obtained. In this case the target-matrices are defined as
\begin{equation}
	\Q_k=\H_k^{\rm{H}}\H_k\in\Cset^{K\times K},\;\;\; \forall k\in\{1,\ldots,K\},
\end{equation}
where $(\cdot)^{\rm{H}}$ denotes the conjugate transpose, the matrices $\H_k\in\Cset^{K\times K}$ denote the (flat fading) channel coefficients from each of the $K$ transmit antennas to each of the $K$ receive antennas of the $k$-th user ($k\in\{1,\ldots,K\}$). The solution matrix $\hB$ in this context contains the desired transmission beamforming
coefficients, such that its $k$-th row contains the coefficients for transmission to the $k$-th user (see \cite{yeredor2012sequentially} for a detailed derivation in this
context).

\begin{figure}
	\centering
	\includegraphics[width=0.5\textwidth]{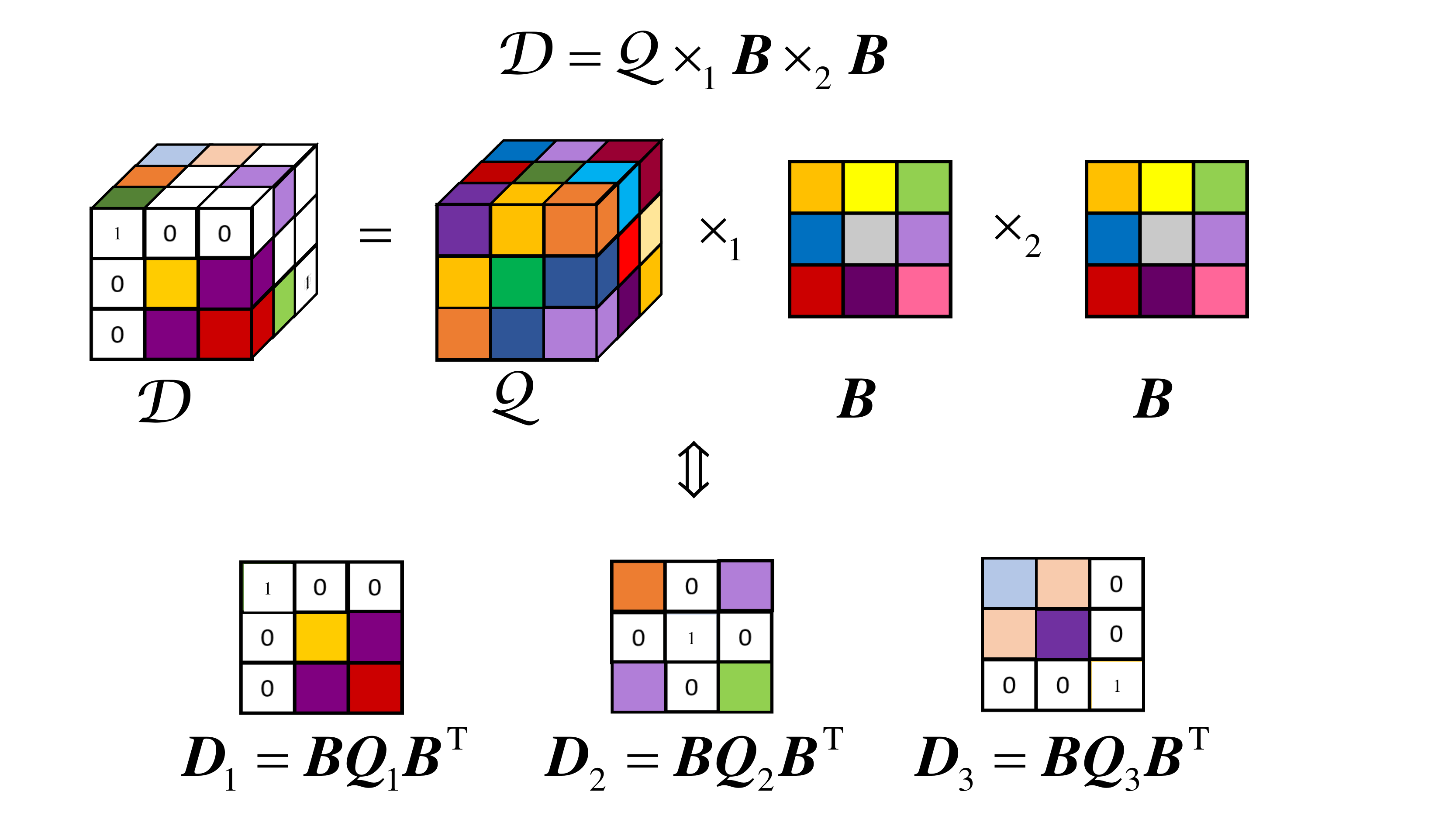}
	\caption{Illustration of SeDJoCo as a tensor decomposition or, equivalently, as a joint matrix transformation. Here, $\times_n$ denotes the $n$-mode product of a tensor with a matrix.}
	\label{fig1tensor}
	\vspace{-0.3cm}
\end{figure}

In recent years there has been a growing interest in the generalization of a single dataset BSS problem to multiple datasets BSS problem, termed Joint Blind Source Separation
(JBSS, \cite{li2009joint,anderson2012joint} and references therein). Examples of applications that motivate the interest in JBSS are the analysis of multi-subject datasets of
electroencephalography (EEG) data (\cite{li2011joint,chatel2013joint}) and functional magnetic resonance imaging (fMRI, \cite{lee2008independent}). A particular case of JBSS is
the extension of the ICA model, termed IVA, which is formulated as follows. Consider $M$ datasets of mixtures
\begin{equation} \label{IVA_model}
	\X^{(m)}=\A^{(m)}\mS^{(m)},\;\;\; \forall m\in\{1,\ldots,M\},
\end{equation}
where $\mS^{(m)}=[\us_1^{(m)}\;\cdots\;\us_K^{(m)}]^{\tps}\in\Rset^{K\times T}$ denotes a matrix of $K$ source signals of length $T$ (for all $m\in\{1,\ldots,M\}$), belonging to
the $m$-th dataset out of $M$ such sets. In each dataset the sources are mixed with an unknown (deterministic) respective mixing-matrix $\A^{(m)}\in\Rset^{K\times K}$, and the
observed mixture signals are given by $\X^{(m)}\in\Rset^{K\times T}$. Based on the observed mixtures datasets $\left\{\X^{(m)}\right\}_{m=1}^{M}$, it is desired to estimate all
$M$ mixing-matrices and thereby recover the source signals. In the same manner as in the standard ICA model, in IVA, too, the sources within each dataset are assumed to be
mutually statistically independent. Clearly, IVA amounts to $M$ independent standard ICA problems when no statistical dependence between source signals across different datasets
exists. However, in IVA statistical dependence between respective sources from different datasets is considered, i.e., the vector $\us_k^{(m_1)}$ may depend on the vector
$\us_k^{(m_2)}$ (for all $m_1,m_2\in\{1,\ldots,M\}$ and all $k\in\{1,\ldots,K\}$), but any two vectors $\us_{k_1}^{(m_1)}$ and $\us_{k_2}^{(m_2)}$ are statistically independent
when $k_1 \neq k_2$ for any $m_1,m_2 \in\{1,\ldots,M\}$. One example which is suitable for this model is group fMRI data analysis, where coherence between estimates of the source
signals across different subjects (i.e., datasets) is exploited for post-analysis of the data, e.g., for group level inference and for the study of inter-subject variability
\cite{calhoun2001method}. Another example is the separation of mixtures of color images \cite{via2011maximum}\addra{. Suppose that we are given a set of $K=2$ linear mixtures of two color images, where each image consists of three color layers (Red, Green and Blue), and the respective layers from each image are mixed separately so as to form the respective layers of the mixed images. So there are $M=3$ mixtures sets, one for each color layer. The two mixed images are independent, but the color layers in each image are usually strongly correlated, giving rise to dependence between sets while maintaining independence within sets.}

As we shall show immediately, it turns out that in the zero-mean Gaussian model the resulting likelihood equations for obtaining the ML estimates of the matrices
$\left\{\B^{(m)}\triangleq {\A^{(m)-1}}\right\}_{m=1}^{M}$ in the IVA problem require a solution of what can be regarded as an extension to the SeDJoCo problem, which we term an
``extended" SeDJoCo problem.

\subsection{Extended SeDJoCo as the Likelihood Equations}
In order to simplify the exposition, we introduce an equivalent formulation for the IVA model \eqref{IVA_model}. Define the block diagonal matrix $\bA \triangleq
\text{Bdiag}\left(\A^{(1)},\ldots,\A^{(M)}\right) \in \Rset^{KM \times KM}$, where the $\text{Bdiag}(\cdot)$ operator creates a block-diagonal matrix from its square matrix
arguments. Additionally, define the matrices $\bS \triangleq \left[{\mS^{(1)\tps}}\;\cdots\;{\mS^{(M)\tps}}\right]^{\tps} \in \Rset^{KM \times T}$ and $\bX \triangleq
\left[{\X^{(1)\tps}}\;\cdots\;{\X^{(M)\tps}}\right]^{\tps} \in \Rset^{KM \times T}$. Model \eqref{IVA_model} can now be more compactly expressed as
\begin{equation} \label{IVA_compact_model}
	\bX=\bA\bS.
\end{equation}
Left-multiplying by $\bB \triangleq \bA^{-1}$ and applying the $\text{vec}(\cdot)$ operator (which concatenates the columns of an $M \times N$ matrix into an $MN \times 1$ vector) we get
\begin{equation} \label{IVA_vec_model}
	\text{vec}(\bS)=(\I_T \otimes \bB)\text{vec}(\bX) \in \Rset^{KMT \times 1},
\end{equation}
where $\I_T$ is the identity matrix of dimension $T$ and $\otimes$ denotes the Kronecker product.

At this point we recall that the term ``blind" usually implies that no information is available regarding the sources, except for their mutual independence within each dataset
(hence, the term IVA). In a ``semi-blind" scenario, more {\it a-priori} structural or statistical information about the sources might be available. This paper addresses the
semi-blind scenario with {\it a-priori} knowledge regarding the joint distribution of the sources. In particular, we assume that the source signals are zero-mean Gaussian with
(known) temporal covariance matrices $\C_k^{(m_1,m_2)} \triangleq E\left[\us_k^{(m_1)} {\us_{k}^{(m_2)\tps}}\right]\in\Rset^{T\times T}$, namely $\C_k^{(m_1,m_2)}$ is the
temporal covariance matrix between the $k$-th source at the $m_1$-th set and the $k$-th source at the $m_2$-th set. \addra{Admittedly, such a scenario may seem too far-fetched in practice. We provide some possible justifications thereto (including a description of a practical example where such prior knowledge may be available) in the next subsection. However, to proceed with the exposition, assume for now that such prior knowledge is indeed available.}

Once the distribution of the sources is known (denoted by $p_{\tiny{\bS}}(\bS)$), the parameterized probability density of the observed mixtures can be expressed as
\begin{equation} \label{pdf_of_x}
	\begin{aligned}
		p_{\tiny{\bX}}(\bX;\bB) &= \left|\text{det}(\I_T \otimes \bB)\right|p_{\tiny{\bS}}\left((\I_T \otimes \bB)\text{vec}(\bX)\right)\\
		&=\left|\text{det}\bB\right|^{T}p_{\tiny{\bS}}\left((\I_T \otimes \bB)\text{vec}(\bX)\right).
	\end{aligned}
\end{equation}
For the explicit expression of the likelihood function of $\bB$ we must obtain an explicit form of $p_{\tiny{\bS}}(\bS)$. Fortunately, things can be simplified by exploiting the
statistical independence between all Source Component Vectors (SCVs), defined as
\begin{equation} \label{SCV_def}
	\bs_k \triangleq \text{vec}(\mS_k^{\addrb{\tps}}),\;\;\; \forall k\in\{1,\ldots,K\},
\end{equation}
where $\mS_k$ is the $k$-th source component matrix, defined as
\begin{equation} \label{SCM_def}
	\mS_k \triangleq \left[\us_k^{(1)}\;\cdots\;\us_k^{(M)}\right]^{\tps} \in \Rset^{M \times T}.
\end{equation}
The covariance matrix of each SCV is given by
\begin{equation} \label{SCM_def}
	\bC_k \triangleq E\left[\bs_k\bs_k^{\tps}\right] =\begin{bmatrix}
		\C_k^{(1,1)} & \cdots & \C_k^{(1,M)}\\
		\vdots & \ddots & \vdots\\
		\C_k^{(M,1)} & \cdots & \C_k^{(M,M)}\end{bmatrix} \in \Rset^{MT \times MT},
\end{equation}
and we denote the respective block-partition of its inverse as
\begin{equation}
	\bC_k^{-1} \triangleq \begin{bmatrix}
		\P_k^{(1,1)} & \cdots & \P_k^{(1,M)}\\
		\vdots & \ddots & \vdots\\
		\P_k^{(M,1)} & \cdots & \P_k^{(M,M)}\end{bmatrix}\addra{\triangleq\bP_k},
\end{equation}
where $\P_k^{(m_1,m_2)} \in \Rset^{T \times T}$, to be used below. Using the Gaussian distribution of the sources we have
\begin{equation} \label{pdf_of_s}
	\begin{aligned}
		p_{\tiny{\bS}}(\bS)&={\displaystyle \prod_{k=1}^{K} p_{\bs_k}(\bs_k)}\\
		&={\displaystyle \prod_{k=1}^{K} (2\pi)^{-\frac{MT}{2}}\left|\text{det}\bC_k\right|^{-\frac{1}{2}}\exp\left(-\frac{1}{2}\bs_k^{\tps}\bC_k^{-1}\bs_k\right)},
	\end{aligned}
\end{equation}
and using
\begin{multline} \label{source_k_from_x}
	\B^{(m)}\X^{(m)}=\mS^{(m)} \Leftrightarrow \ue_{k}^{\tps}\B^{(m)}\X^{(m)}={\us_{k}^{(m)\tps}},\\
	\forall k\in\{1,\ldots,K\},\forall m\in\{1,\ldots,M\},
\end{multline}
we obtain the normalized log-likelihood of $\bB$, given by
\begin{equation} \label{likelihood_of_B}
	\begin{aligned}
		&\mathcal{L}(\bB) = \mathcal{L}(\B^{(1)},\ldots,\B^{(M)}) \triangleq \frac{1}{T}\log p_{\tiny{\bX}}(\bX;\bB) \\
		& = \gamma+\log\left|\text{det}\bB\right|\\
		&-\frac{1}{2T}{\displaystyle \sum_{k=1}^{K}\sum_{\substack{m_1=1\\m_2=1}}^{M}\ue_k^{\tps}\B^{(m_1)}\X^{(m_1)}\P_{k}^{(m_1,m_2)}{\X^{(m_2)\tps}}{\B^{(m_2)\tps}}\ue_k}\\
		& = \gamma+\sum_{\ell=1}^{M}\log\left|\text{det}\B^{(\ell)}\right|\\
		&-\frac{1}{2}{\displaystyle \sum_{k=1}^{K}\sum_{\substack{m_1=1\\m_2=1}}^{M}\ue_k^{\tps}\B^{(m_1)}\Q_{k}^{(m_1,m_2)}{\B^{(m_2)\tps}}\ue_k},
	\end{aligned}
\end{equation}
where the matrices $\Q_{k}^{(m_1,m_2)}$ (to be later referred to as the target-matrices), are defined as
\begin{multline} \label{target_matrices_definition}
	\Q_{k}^{(m_1,m_2)} \triangleq \frac{1}{T}\X^{(m_1)}\P_{k}^{(m_1,m_2)}{\X^{(m_2)\tps}},\\
	\forall k\in\{1,\ldots,K\},\; \forall m_1,m_2\in\{1,\ldots,M\},
\end{multline}
and where $\gamma$ is a constant independent of $\bB$. Since the ML estimate $\hB_{\text{ML}}^{(1)},\ldots,\hB_{\text{ML}}^{(M)}$ of $\B^{(1)},\ldots,\B^{(M)}$ are the global
maximizers of the likelihood function, we seek the solution of
\begin{equation} \label{derivative_of_likelihood}
	\frac{\partial \mathcal{L}(\bB)}{\partial \B^{(m)}} \mathrel{\stackon[1pt]{$=$}{$\scriptstyle!$}} \textrm{O},\;\;\;\forall m\in\{1,\ldots,M\},
\end{equation}
\addra{(where $\mathrel{\stackon[1pt]{$=$}{$\scriptstyle!$}}$ denotes a demand for equality)} which corresponds to the global maximum, where $\textrm{O} \in \Rset^{K \times K}$ is the all-zeros matrix. Indeed, differentiating $\mathcal{L}(\bB)$ w.r.t. $\B^{(m)}$ and
equating to zero yields  (see Appendix \ref{appendix_a} for details)
\begin{equation} \label{derivative_of_likelihood_2}
	\begin{aligned}
		&\left.{\frac{\partial \mathcal{L}(\bB)}{\partial \B^{(m)}}}\right|_{\widehat{\bB}_{\text{ML}}} = \frac{\partial}{\partial \B^{(m)}}\left(\sum_{\ell=1}^{M}\log\left|\text{det}\B^{(\ell)}\right|\right.\\
		&\left.\left.-\frac{1}{2}{\displaystyle \sum_{k=1}^{K}\sum_{\substack{m_1=1\\m_2=1}}^{M}\ue_k^{\tps}\B^{(m_1)}\Q_{k}^{(m_1,m_2)}{\B^{(m_2)\tps}}\ue_k}+\gamma\right)\right|_{\widehat{\bB}_{\text{ML}}}\\
		& = {\hA_{\text{ML}}}^{(m)\tps}-\sum_{k=1}^{K}\sum_{\ell=1}^{M}\E_{kk}\hB_{\text{ML}}^{(\ell)}\Q_{k}^{(\ell,m)} \mathrel{\stackon[1pt]{$=$}{$\scriptstyle!$}} \textrm{O}, \;\;\; \forall m\in\{1,\ldots,M\},
	\end{aligned}
\end{equation}
where $\left\{\hA_{\text{ML}}^{(m)}\triangleq\hB_{\text{ML}}^{(m)-1}\right\}_{m=1}^M$ and $\E_{ij}\triangleq\ue_i\ue_j^{\tps}$. So
\begin{equation} \label{almost_e_sedjoco}
	\sum_{k=1}^{K}\sum_{\ell=1}^{M}\E_{kk}\hB_{\text{ML}}^{(\ell)}\Q_{k}^{(\ell,m)}={\hA_{\text{ML}}}^{(m)\tps},\;\;\;\forall m\in\{1,\ldots,M\}.
\end{equation}
Transposing and left-multiplying by $\hB_{\text{ML}}^{(m)}$ we get
\begin{equation} \label{e_sedjoco}
	\begin{aligned}
		& \sum_{k=1}^{K}\sum_{\ell=1}^{M}\hB_{\text{ML}}^{(m)}\Q_{k}^{(m,\ell)}{\hB_{\text{ML}}}^{(\ell)\tps}\E_{kk}=\I_K,  \;\;\;\forall m\in\{1,\ldots,M\}\\
		&\quad\quad\quad\quad\quad\quad\quad\quad\quad\quad\quad\quad\Leftrightarrow \\
		& \sum_{\ell=1}^{M}\hB_{\text{ML}}^{(m)}\Q_{k}^{(m,\ell)}{\hB_{\text{ML}}}^{(\ell)\tps}\ue_k \triangleq \D_{k}^{(m)}\ue_k = \ue_k,\\
		&\quad\quad\quad\quad\quad\quad\quad\quad\forall m\in\{1,\ldots,M\},\forall k\in\{1,\ldots,K\},
	\end{aligned}
\end{equation}
where we have used $\Q_{k}^{(m_1,m_2)}={\Q_{k}^{(m_2,m_1)\tps}}$. The set of equations \eqref{e_sedjoco} (and, equivalently, \eqref{almost_e_sedjoco}) constitutes the {\it
	extended SeDJoCo} problem. It is easy to see that for the particular case where $M=1$, the problem boils down to the standard (single dataset) SeDJoCo problem
\cite{yeredor2012sequentially}. We emphasize that like in standard SeDJoCo, this means that the $k$-th column of the matrix $\D_{k}^{(m)}$ should be ``drilled" (namely be
all-zeros except for its $k$-th element equaling $1$), but unlike standard SeDJoCo, here the matrix $\D_{k}^{(m)}$ is not necessarily symmetric, in general, so the extended
SeDJoCo transformation is not symmetric in the sense that the rows of $\D_{k}^{(m)}$ are generally not ``drilled". In other words, $\ue_k$ is a right-eigenvector of $\D_{k}^{(m)}$
(corresponding to the eigenvalue $1$), whereas $\ue_k^{\tps}$ is not necessarily a left-eigenvector of $\D_{k}^{(m)}$.

When the source signals in different datasets are statistically independent, the matrices $\bC_k$, as well as their inverses, become block-diagonal, so that all
$\P_{k}^{(m_1,m_2)}$ (and therefore also all $\Q_k^{(m_1,m_2)}$) vanish for all $m_1 \ne m_2$. As a result, the problem reduces to a set of $M$ ``standard" (decoupled) SeDJoCo
problems, since only one element (corresponding to $\ell=m$) is left in each sum in \eqref{e_sedjoco}. However, in a ``true" IVA setup sources from different datasets \textit{are}
correlated, giving rise to a non-degenerate extended SeDJoCo problem.

\addra{\subsection{Justification of the Semi-Blindness Assumption}}

\addra{As already mentioned, our semi-blind scenario, which assumes prior knowledge of the full SCV covariance matrices (in addition to the Gaussianity assumption), may seem questionable. There are, however, several possible arguments in support of considering such a scenario - including a specific practical example.}

\begin{itemize}
	
	\addra{\item From a theoretical point of view: ML estimation always assumes prior knowledge of the full statistical model of the observations, up to the unknown parameters (to be estimated). Such knowledge is not always realistic in practice, but still ML estimation enjoys tremendous popularity as a baseline theoretical approach. Thus, one of our objectives is to derive the ML estimation in the IVA context with a general Gaussian model, and to show how prior knowledge of the SCVs' covariances (which is sufficient in the Gaussian case) can be exploited in an optimal manner if and when it is available.}
	
	\addra{\item Assuming a ``training" period: In some applications the user might have access to the unmixed source signals during some ``training" period, before actually observing the mixtures in the ``operational" period. Assuming that the statistical properties of the sources during the training period remain valid during the operational period, these properties may be estimated during the training period, to be used in turn during the operational period.}
	
	\addra{\item A possible iterative scheme: When the covariance matrices are not known {\it a-priori}, but can be succinctly parameterized (e.g., in the case of stationary parametric auto-regressive / moving average sources), an iterative separation strategy may be used as follows: The sources are first estimated by some initial (non-ML) separation (if possible), which may not be optimal, but would still be reasonable enough to allow subsequent estimation of the required parameters for the covariance matrices from the separated signals. Then a semi-blind framework would be applied using the estimated covariance matrices, possibly with successive refinements by repeating the process - thereby approaching asymptotic optimality. In fact, it is also conceivable to operate such a process in an adaptive (rather than a batch) scheme, in which the estimation of the covariance matrices is interlaced with the estimation of the unmixing matrices as new measurements keep flowing in. However, this is a topic for further research, beyond the scope of the current paper.}
	
	\addra{\item A ``real world" practical example: Consider $K$ different sources (e.g., transmitters), remotely positioned at known locations, transmitting different, independent signals, which all have the same (known) spectrum and can be assumed Gaussian. This is a very common assumption, e.g., when the transmitters transmit different communication signals with the same standard modulation, which has a known spectrum (implying a known autocorrelation function). \addrb{For example, Orthogonal Frequency Division Multiplexing (OFDM) digital communication signals are commonly modeled as Gaussian, see, e.g., \cite{banelli2000theoretical}}. Assume first that these signals are received at a single site by $K$ different sensors (antennas) with different (unknown, or uncalibrated) radiation patterns, directed at different directions (not necessarily in the form of a calibrated phased array). The signal received at each antenna is then a different (unknown) linear mixture of all $K$ sources. Note that the mixture is non-separable (without further information on the signals) since it consists of Gaussian sources with identical spectra.
		
		Now assume a second, distant reception site with $K$ similar antennas. The source signals received at this site are differently-delayed versions of those from the other site, due to propagation delays. Knowing the positions of the sources and of the sensor sites, these delays can be calculated, and can be readily used to obtain the cross-correlations between respective sources at both sites, which is simply their autocorrelations, shifted by the respective (positive or negative) delay differences. Thus, the $2K$ signals at the two sites give rise to an IVA problem with $M=2$ sets (easily extended to any $M$ by adding more sites), in a semi-blind Gaussian scenario: all necessary covariance (and cross-covariance) matrices are known. Moreover, the signals at each site alone are non-separable - yet, using our semi-blind IVA scheme we can obtain optimal separation of the sources.
		
		The fine details of this scenario introduce some minor complications (e.g., in the context of communication signals, a complex-valued extension of our results should be used), so we shall not pursue this problem in here any further - however we did get good separation results with this scenario (even in the presence of additive noise), so this is at least one practical example where prior knowledge of the required covariance matrices is quite realistic.}
	
\end{itemize}

Interestingly, \addra{we may also add that} a very similar formulation \addra{to the extended SeDJoCo} can again be linked to a CBF problem, but in an extended multi-cast framework, where $M$ inter-connected transmitters, each with $K$ antennas and $K$ associated users are required to attain perfect interference cancellation transmission.
This scenario is described in \cite{cheng2015extension} in detail along with the mathematical derivation which leads to the following set of equations
\begin{multline} \label{e_sedjoco_2}
	\sum_{\ell=1}^{M}\B^{(\ell)}\Q_k^{(\ell,m)}{\B^{(m)\rm{H}}}\ue_k\triangleq\tD_k^{(m)}\ue_k=\alpha\ue_k,\\
	\forall k\in\{1,\dots,K\} \; \forall m\in\{1,\ldots,M\},
\end{multline}
where in this case the target-matrices are defined as
\begin{multline}\label{CBF_target_mat}
	\Q_k^{(m_1,m_2)}\triangleq{\H_k^{(m_2,m_1)\rm{H}}}\H_k^{(m_2,m_2)}\in\Cset^{K\times K},\\
	\forall k \in \{1,\ldots,K\} \; \forall m_1,m_2 \in \{1,\ldots,M\}\addra{,}
\end{multline}
\addra{where each matrix $\H_k^{(m_1,m_2)}\in\Cset^{K\times K}$ contains the flat fading channel coefficients from the $K$ antennas of the $m_2$-th transmitter to the $K$ antennas of the $k$-th associated user of the $m_1$-th transmitter.} Evidently, these equations resemble the extended SeDJoCo equations \eqref{e_sedjoco} (with a change of the summation index; with a conjugate transpose replacing the transpose; and
with the allowed scaling factor $\alpha$). The solution of the resulting extended SeDJoCo problem would enable interference free delivery of the intended data streams to all
users. Note that the only main difference of \eqref{e_sedjoco_2} from \eqref{e_sedjoco} is captured by the (different) definitions of the transformed matrices $\tD_k^{(m)}$ and
$\D_k^{(m)}$.

\delra{The appearance of extended SeDJoCo (up to minor differences) in these two different contexts of IVA and CBF, along with its special structure, and the belief that it may appear in
	other applications employing joint matrix transformations, are}
\addra{These arguments serve as} the basis for our motivation to delve deeper into this problem and \addra{to} provide further results and insights regarding
these equations.

The rest of this paper is structured as follows. In Section \ref{sec_Theory} we consider theoretical aspects of the problem like an alternative formulation, conditions for the
existence of a solution and discussion of its non-uniqueness. In Section \ref{sec_iCRLB} we derive the induced Cram\'er-Rao lower bound (iCRLB, \cite{yeredor2010blind}) on the
interference-to-signal ratio (ISR) for the underlying Gaussian IVA problem. In Section \ref{sec_Solution} we propose two iterative solution algorithms for extended SeDJoCo.
Comparative simulation results are presented in Section \ref{sec_Simulation}, and Section \ref{sec_Conclusion} is devoted to conclusions.

\section{Theoretical Aspects of Extended SeDJoCo}
\label{sec_Theory}
\subsection{Problem Formulations, Existence, Non-Uniqueness}
We start with (re)formulating the (context-free) extended SeDJoCo problem. Recall that in this problem we consider $M$ sets, giving rise to $KM^2$ target-matrices, each of
dimension $K \times K$, with a solution in the form of $M$ $K \times K$ matrices. Hence, the extended SeDJoCo is stated as follows: Given $KM^2$ target-matrices
$\left\{\Q_k^{(m_1,m_2)}\right\} \; \forall k \in \{1,\ldots,K\}, \forall m_1,m_2 \in \{1,\ldots,M\}$,

\noindent \textit{P1: find a set of $M$ $K \times K$ matrices $\left\{\B^{(m)}\right\}_{m=1}^M$, such that}
\begin{multline}\label{extended_sedjoco_original}
	\left[\sum_{\ell=1}^{M}{\B^{(m)}\Q_k^{(m,\ell)}{\B^{(\ell)\tps}}}\right]\ue_k \triangleq \D_{k}^{(m)}\ue_k =\ue_k,\\
	\forall k \in \{1,\ldots,K\}, \forall m \in \{1,\ldots,M\}.
\end{multline}
The meaning of this statement is that the transformed matrices $\D_k^{(m)}$ should all be \textit{exactly} ``drilled" in their $k$-th column.

Equivalently, this problem can be stated as:

\noindent \textit{P2: find a set of $KM$ vectors $\left\{\ub_k^{(m)} \in \Rset^{K \times 1} \right\}, k \in \{1,\ldots,K\}, m \in \{1,\ldots,M\}$, such that}
\begin{multline}\label{extended_sedjoco_alternative}
	\sum_{\ell=1}^{M}{{\ub_{k_1}^{(m)\tps}}\Q_{k_2}^{(m,\ell)}}\ub_{k_2}^{(\ell)}=\delta_{k_1k_2},\\
	\forall k_1,k_2 \in \{1,\ldots,K\}, \forall m \in \{1,\ldots,M\},
\end{multline}
\textit{where $\delta_{k_1k_2}$ denotes Kronecker's delta function (which is $1$ if $k_1=k_2$ and $0$ otherwise).}

This formulation suggests that each solution vector $\ub_k^{(m)}$ (which is simply the $k$-th row of $\B^{(m)}$) is orthogonal to all (transformed) vectors $\ubeta_{k_2}^{(m)}
\triangleq \sum_{\ell=1}^{M}\Q_{k_2}^{(m,\ell)}\ub_{k_2}^{(\ell)}$ where $k_1 \ne k_2$ (for any $m \in \{1,\ldots,M\}$). With these notations, we have that ${\B^{(m)\tps}} =
\left[\ub_1^{(m)}\;\cdots\;\ub_K^{(m)}\right]$ for all $m \in \{1,\ldots,M\}$.

As seen from both formulations above, the extended SeDJoCo requires the solution of $MK^2$ equations in $MK^2$ unknowns, the elements of the matrices
$\left\{\B^{(m)}\right\}_{m=1}^{M}$. Nevertheless, since these equations are nonlinear (in particular, they contain only $2^{\text{nd}}$ order monomials of the unknowns),
conclusions regarding the existence and/or uniqueness of the solution are non-trivial. The following is a sufficient condition for the existence of a (generally non-unique)
solution. In the sequel we shall characterize a set of solutions which may exist when at least one solution exists.

\begin{thm}[a sufficient condition for existence of a solution] \label{theorem1}
	For a given set of target-matrices $\left\{\Q_k^{(m_1,m_2)}\right\} \in \Rset^{K \times K} \; \forall k \in \{1,\ldots,K\}, \forall m_1,m_2 \in \{1,\ldots,M\}$, a solution for the
	associated extended SeDJoCo problem exists if all $K$ matrices $\left\{\Omeg_k\right\}_{k=1}^{K}$, defined as
	\begin{equation} \label{augmented_taget}
		\Omeg_k \triangleq \begin{bmatrix}
			\Q_k^{(1,1)} & \cdots & \Q_k^{(1,M)}\\
			\vdots & \ddots & \vdots\\
			\Q_k^{(M,1)} & \cdots & \Q_k^{(M,M)}\end{bmatrix} \in \Rset^{KM \times KM},
	\end{equation}
	are Positive Definite (PD).
\end{thm}

\begin{proof}
	Let $\Omeg_1,\ldots,\Omeg_K$ denote a set of (symmetric, real-valued) PD matrices, constructed from the target-matrices as defined in \eqref{augmented_taget}, and let
	$\lambda_k>0$ denote the smallest eigenvalue of $\Omeg_k$, $k=1,\ldots,K$. Define
	\begin{equation} \label{augmented_solution_matrix}
		\tB \triangleq \left[\B^{(1)}\;\cdots\;\B^{(M)}\right] \in \Rset^{K \times KM},
	\end{equation}
	and denote its rows as $\tb_k^T, k=1,\ldots,K$. Now consider the function
	\begin{equation}
		\label{realC}
		C\left(\tB\right)\triangleq\sum_{m=1}^{M}\log\left|\det\B^{(m)}\right|-\frac{1}{2}\sum_{k=1}^K\tb_k^{\tps}\Omeg_k\tb_k.
	\end{equation}
	For all nonsingular $\B^{(1)},\dots,\B^{(M)}$, $C\left(\tB\right)$ is obviously a continuous and differentiable function of all elements of all the matrices.
	In addition, $C\left(\tB\right)$ is bounded from above:
	\begin{equation}
		\label{Cle}
		\begin{split}
			C(\tB)&=\sum_{m=1}^{M}\log\left|\det\B^{(m)}\right|-\frac{1}{2}\sum_{k=1}^K\tb_k^{\tps}\Omeg_k\tb_k\\
			&\le \sum_{m=1}^{M}\log\prod_{k=1}^K\|\ub_k^{(m)}\|-\frac{1}{2}\sum_{k=1}^K\lambda_k\tb_k^{\tps}\tb_k\\
			&=\frac{1}{2}\sum_{m=1}^{M}\sum_{k=1}^K\left\{\log\|\ub_k^{(m)}\|^2-\lambda_k\|\ub_k^{(m)}\|^2\right\}\\
			&\le\frac{1}{2}\sum_{m=1}^{M}\sum_{k=1}^K\left\{-\log\lambda_k-1\right\}\\
			&=\frac{M}{2}\sum_{k=1}^K\left\{\log\lambda_k-1\right\},
		\end{split}
	\end{equation}
	where $||\cdot||$ denotes the $\ell^2$-norm, and where we have used the properties
	\begin{enumerate}
		\item $|\det\B^{(m)}|\le\prod_{k=1}^K\|\ub_k^{(m)}\|$ (Hadamard's inequality);
		\item $\tb_k^{\tps}\Omeg_k\tb_k\ge\lambda_k\|\tb_k\|^2$;
		\item $\|\tb_k\|^2 = \sum_{m=1}^M\|\ub_k^{(m)}\|^2$; and
		\item $\log x-\lambda x\le -\log\lambda-1$ for all $x>0$.
	\end{enumerate}
	
	Note also that $C\left(\tB\right)$ tends to $-\infty$ when at least one of the matrices $\left\{\B^{(m)}\right\}_{m=1}^{M}$ approaches any singular matrix,
	and that, in addition, $C\left(\tB\right)$ has the property
	\begin{equation}
		C\left(\alpha\cdot\tB\right)\xrightarrow{\alpha\rightarrow\infty} -\infty,\;\;\; \forall\tB\in\Rset^{K \times KM}.
	\end{equation}
	Consequently, $C\left(\tB\right)$ must attain a maximum for (at least) some set of nonsingular matrices $\left\{\B^{(m)}\right\}_{m=1}^{M}$. Being a smooth function thereof
	derivative w.r.t. each $\left\{\B^{(m)}\right\}_{m=1}^{M}$ at the maximum point must vanish.
	
	Indeed, differentiating $C\left(\tB\right)$ w.r.t. $\B^{(m)}$ we get 
	\begin{equation}
		\label{dCdB}
		\begin{split}
			&\frac{\partial C\left(\tB\right)}{\partial \B^{(m)}}={\A^{(m)\tps}}-\frac{1}{2}\frac{\partial}{\partial \B^{(m)}}\left[\sum_{k=1}^K\ue_k^{\tps}\tB\Omeg_k\tB^{\tps}\ue_k\right]\\
			&={\A^{(m)\tps}}-\frac{1}{2}\frac{\partial}{\partial \B^{(m)}}\left[\sum_{k=1}^K\sum_{\ell=1}^M\sum_{p=1}^M\ue_k^{\tps}\B^{(\ell)}\Q_k^{(\ell,p)}{\B^{(p)\tps}}\ue_k\right]\\
			&={\A^{(m)\tps}}-\frac{1}{2}\sum_{k=1}^K\sum_{\ell=1}^M2\E_{kk}\B^{(\ell)}\Q_k^{(\ell,m)}\\
			&={\A^{(m)\tps}}-\sum_{k=1}^K\sum_{\ell=1}^M\E_{kk}\B^{(\ell)}\Q_k^{(\ell,m)}, \; \forall m\in\{1,\ldots,M\},
		\end{split}
	\end{equation}
	where $\A^{(m)}\triangleq {\B^{(m)-1}}$, and we have used the properties stated in Appendix \ref{appendix_a}. And so, equating to zero, transposing and left-multiplying by
	$\B^{(m)}$, we get
	\begin{equation}
		\label{result_e_sedjoco}
		\sum_{k=1}^K\sum_{\ell=1}^M\B^{(m)}\Q_k^{(m,\ell)}{\B^{(\ell)\tps}}\E_{kk}=\I_K,  \;\;\;\forall m \in \{1,\ldots,M\}
	\end{equation}
	which implies \eqref{extended_sedjoco_original}. This means that a solution of extended SeDJoCo must exist as the maximizer of $C\left(\tB\right)$, as long as
	$\left\{\Omeg_k\right\}_{k=1}^K$ are all PD. \vspace*{-0.5cm}
	\begin{align}
		\tag*{$\square$}
	\end{align}
\end{proof}
\vspace*{-0.2cm}
Note that this general result holds for any extended SeDJoCo problem, and is not limited to the context of IVA. 
\addra{For an IVA problem, when the $\Q_k^{(m_1,m_2)}$ matrices are given by \eqref{target_matrices_definition}, it is easy to show that if all the SCVs covariance matrices $\bC_k$ are PD and finite, so are all $\Omeg_k$, with probability $1$ (w.p.1). To observe this, note that in this case each $\Omeg_k$ can be expressed as $\Omeg_k=\Chi\bP_k\Chi^\tps$, where $\Chi\triangleq\text{Bdiag}\left(\X^{(1)},\ldots,\X^{(M)}\right)\in\Rset^{KM\times MT}$. If $\bC_k$ is PD and finite, so is $\bP_k$, and since $\Chi$ is full rank w.p.1, $\Omeg_k$ is also PD (w.p.1). It is important to realize, however, that existence of an extended SeDJoCo solution in a given IVA problem does not necessarily imply separability, because an infinite number of solutions may exist over a manifold in the parameters space. In fact, this is what happens when the identifiability conditions stated in \cite{anderson2014independent} are not satisfied, e.g., when two Gaussian SCVs have the same covariance matrices – in which case the iCRLB is infinite for the associated ISRs. When the identifiability conditions in \cite{anderson2014independent} are satisfied, then although the SeDJoCo solution is still not unique (see below), the multiple solutions are all isolated (w.p.1), and only one corresponds to the global maximum of the likelihood function}.

Having addressed the issue of existence of a solution, we proceed to discuss the issue of uniqueness. \addrb{Note first, that since, as mentioned above, extended SeDJoCo is a system of $MK^2$ equations in $MK^2$ unknowns (elements of $\tB$), where each equation is a second degree multinomial in the unknowns, B\'{e}zout's theorem (e.g., \cite{gibson2001elementary}) asserts that there are at most $2^{MK^2}$ distinct real-valued solutions. Indeed, according to our experience, when a solution exists, it is not unique, in general. Moreover, we can generally characterize $K!$ essentially different solutions whenever a single solution exists.
	
	To this end, assume an extended SeDJoCo problem associated with a set of target-matrices $\left\{{\Q_{k}^{(m_1,m_2)}}\right\}$ that satisfy the existence condition (all implied $\Omeg_k$ matrices are PD). Denote by $\left\{\B^{(m)}\right\}_{m=1}^{M}$ a solution to this problem. Now define a new set of target-matrices $\left\{\brvQ_k^{(m_1,m_2)}\right\}$, such that  ($\forall
	m_1,m_2\in\{1,\ldots,M\}$)
	\begin{equation}\label{newsetfirsttwo}
		\brvQ_1^{(m_1,m_2)}=\Q_2^{(m_1,m_2)} \; , \; \brvQ_2^{(m_1,m_2)}=\Q_1^{(m_1,m_2)},
	\end{equation}
	and
	\begin{equation}\label{newset}
		\brvQ_k^{(m_1,m_2)}=\Q_k^{(m_1,m_2)} \quad \forall k \in \{3,\ldots,K\},
	\end{equation}
	thereby defining a modified (permuted) extended SeDJoCo problem. Obviously, the modified problem also has at least one solution (the existence condition is still satisfied), but a solution of the original problem generally does not solve this modified problem. However, an iterative algorithm starting at the solution $\left\{\B^{(m)}\right\}_{m=1}^{M}$ of the original problem is likely to reach an essentially different solution of the modified problem, in which all elements of the resulting solution matrices are generally different from all elements of the solution matrices of the original problem, because different coefficients now multiply different products of the unknowns in the system of equations. Thus, let us denote by $\left\{\brvB^{(m)}\right\}_{m=1}^{M}$ the resulting solution of the modified problem.
	
	Now let $\PI_{1,2} \in \Rset^{K \times K}$ denote the (symmetric) permutation matrix that swaps the first and second elements of a vector,
	namely $\PI_{1,2}\ue_1=\ue_2$, $\PI_{1,2}\ue_2=\ue_1$ and $\PI_{1,2}\ue_k=\ue_k$ for all other $k \in \{3,\ldots,K\}$. Consider the set of matrices
	$\left\{\B^{'(m)}\triangleq\PI_{1,2}\brvB^{(m)}\right\}$ for all $m \in \{1,\ldots,M\}$. We assert that this set of matrices solves the extended SeDJoCo problem induced by the
	\textit{original} set $\left\{\Q_{k}^{(m_1,m_2)}\right\}$, since
	\begin{equation}\label{newsetfirsttwo}
		\begin{split}
			&\left[\sum_{m=1}^{M}\B'^{(\ell)}\Q_1^{(\ell,m)}{\B'^{(m)\tps}}\right]\ue_1\\
			&=\PI_{1,2}\left[\sum_{m=1}^{M}\brvB^{(\ell)}\brvQ_2^{(\ell,m)}{\brvB^{(m)\tps}}\right]{\PI_{1,2}}^{\tps}\ue_1\\
			&=\PI_{1,2}\left[\sum_{m=1}^{M}\brvB^{(\ell)}\brvQ_2^{(\ell,m)}{\brvB^{(m)\tps}}\right]\ue_2\\
			&=\PI_{1,2}\ue_2=\ue_1, \;\;\; \forall \ell \in \{1,\ldots,M\},
		\end{split}
	\end{equation}
	and, similarly, $\left[\sum_{m=1}^{M}\B'^{(\ell)}\Q_2^{(\ell,m)}{\B'^{(m)\tps}}\right]\ue_2=\ue_2$, and, of course,
	$\left[\sum_{m=1}^{M}\B'^{(\ell)}\Q_k^{(\ell,m)}{\B'^{(m)\tps}}\right]\ue_k=\ue_k$ for all other $k \in \{3,\ldots,K\}$ (for all $\ell \in \{1,\ldots,M\}$). This means that \addra{in addition to $\left\{\B^{(m)}\right\}_{m=1}^M$ there exists} \delra{we have found} an additional solution \addra{$\left\{\B'^{(m)}\right\}_{m=1}^M$} to the original extended SeDJoCo problem, \addra{which is} \delra{via} a permutation of an ``essentially differernt" solution \addra{$\left\{\brvB^{(m)}\right\}_{m=1}^M$} to a permuted extended SeDJoCo problem. Since any permutation matrix can be expressed as the product of two-elements-permutation matrices, we may generalize the above result, i.e., there may exist, in general, $K!$ (the number of possible permutations) such different solutions for a $K$-dimensional extended SeDJoCo problem.
	
	Note that, strictly speaking, we did not {\it prove} that the permuted extended SeDJoCo problem yields an {\it essentially different} solution, since theoretically the resulting solution of the permuted problem may just be a permuted version of the solution to the original problem. However, based on our empirical experience, we conjecture that with randomly generated target matrices (such as in IVA), an iterative algorithm starting at a solution of the original problem would ``almost surely" reach an essentially different (not just permuted) solution of the permuted problem. }\addra{In the context of our IVA problem, all these $K!$ different solutions would be local maxima of the Likelihood function, but (w.p.1) only one of them would correspond to the global maximum, and may be found using strategies such as those advocated (in the context of ICA) in \cite{yeredor2016multiple, weiss2017maximum}.}

\section{iCRLB on the ISR for JBSS}
\label{sec_iCRLB}
The ISR is a common measure in BSS which quantifies the ``quality" of separation. More specifically, by definition
\begin{multline} \label{ISR_defenition}
	\text{ISR}^{(m)}_{\mathhlra{_{k\ell}}} \triangleq E\left[\frac{\left|\left(\hB^{(m)}\A^{(m)}\right)_{\mathhlra{_{k\ell}}}\right|^2}{\left|\left(\hB^{(m)}\A^{(m)}\right)_{\mathhlra{_{kk}}}\right|^2}\right]\cdot \frac{E\left[\us^{(m)\tps}_{\mathhlra{_\ell}}\us^{(m)}_{\mathhlra{_\ell}}\right]}{E\left[\us^{(m)\tps}_{\mathhlra{_k}}\us^{(m)}_{\mathhlra{_k}}\right]},\\
	\forall \mathhlra{k,\ell} \in \{1,\ldots,K\}, \mathhlra{k} \neq \mathhlra{\ell}, \forall m \in \{1,\ldots,M\},
\end{multline}
measures the expected relative residual energy of the \addra{$\ell$-th} source in the reconstruction of the \addra{$k$-th} source in the \addra{$m$-th} dataset.

\addra{By deriving the CRLB on the estimation of the mixing matrix and using the equivariance property, it is possible to obtain the iCRLB on the ISR (e.g., \cite{doron2007cramer, comon2010handbook, yeredor2010blind, anderson2014independent}). In \cite{anderson2014independent} general expressions for the iCRLB are provided in the context of a general IVA problem. It is shown that
	\begin{equation}
		\label{eq:ISR_Ka}
		\text{ISR}_{k\ell}^{(m)}\ge \frac{1}{T}\cdot\ue_m^{\tps}\left[\Ka_{k\ell}-\Ka_{\ell k}^{-1}\right]^{-1}\ue_m\cdot\frac{\Tr\left(\C_\ell^{(m,m)}\right)}{\Tr\left(\C_k^{(m,m)}\right)},
	\end{equation}
	where $\Tr(\cdot)$ denotes the trace operator, and where the elements of the matrices $\Ka_{k\ell}\in\Rset^{M\times M}$ are defined as
	\begin{equation}
		\label{eq:defKa}
		\mathcal{K}_{k\ell}[m,n]\triangleq \frac{1}{T}\Tr\left(\Ga_k^{(n,m)}\C_\ell^{(m,n)}\right),
	\end{equation}
	in which the matrices $\Ga_k^{(m,n)}\in\Rset^{T\times T}$ are defined as
	\begin{equation}
		\label{eq:defGa}
		\Ga_k^{(m,n)}\triangleq E\left[\uphi_k^{(m)\tps}\uphi_k^{(n)}\right].
	\end{equation}
	Here $\uphi_k^{(m)}\in\Rset^{1\times T}$ denotes the score vector of the $m$-th component of the $k$-th source, namely the derivative of the (negative) log of the probability distribution of the $k$-th SCV $\bs_k$ w.r.t.\ its $m$-th component $\us_k^{(m)}$. Different IVA models naturally have different score vectors, giving rise to different $\Ga_k^{(m,n)}$ matrices and, thereby, to different iCRLBs. In order to obtain an explicit expression for the iCRLB in a given IVA model, the respective $\Ga_k^{(m,n)}$ matrices need to be explicitly calculated. In \cite{anderson2014independent} the authors derive explicit results for the simple case of temporally independent, identically distributed (i.i.d.) Gaussian sources. To obtain $\Ga_k^{(m,n)}$ for our more general temporal models, recall that
	\begin{multline}
		-\log p_{\bs_k}(\bs_k)=\tfrac{1}{2}\log\text{det}|2\pi\bC_k|+\tfrac{1}{2}\bs_k^{\tps}\bC_k^{-1}\bs_k\\
		=\tfrac{1}{2}\log\text{det}|2\pi\bC_k|+\tfrac{1}{2}\bs_k^{\tps}\bP_k\bs_k,
	\end{multline}
	so that its derivative w.r.t.\ the entire SCV $\bs_k$ is given by $\bs_k^{\tps}\bP_k\in\Rset^{1\times MT}$. Its $m$-th component is therefore given by
	\begin{equation}
		\uphi_k^{(m)}=\sum_{p=1}^M \us_k^{(p)\tps}\P_k^{(p,m)}.
	\end{equation}
	Substituting in \eqref{eq:defGa} we get
	\begin{multline}
		\Ga_k^{(m,n)}=E\left[\uphi_k^{(m)\tps}\uphi_k^{(n)}\right]=E\left[\sum_{p,q=1}^M\P_k^{(m,p)}\us_k^{(p)}\us_k^{(q)\tps}\P_k^{(q,n)}\right]\\
		=\sum_{p,q=1}^M\P_k^{(m,p)}\C_k^{(p,q)}\P_k^{(q,n)}=\sum_{q=1}^M\I_{MT}^{(m,q)}\P_k^{(q,n)}=\P_k^{(m,n)},
	\end{multline}
	where $\I_{MT}$ denotes the $MT\times MT$ identity matrix and $\I_{MT}^{(m,q)}$ denotes its $(m,q)$-th $T\times T$ block (which is $\I_T$ for $q=m$ and all zeros otherwise).
	To conclude, we can now substitute this result into \eqref{eq:defKa}, obtaining
	\begin{equation}
		\mathcal{K}_{k\ell}[m,n]= \frac{1}{T}\Tr\left(\P_k^{(n,m)}\C_\ell^{(m,n)}\right),
	\end{equation}
	which can in turn be substituted into \eqref{eq:ISR_Ka} to yield the iCRLB.}

\section{Solving Extended SeDJoCo}
\label{sec_Solution} To the best of our knowledge, the extended SeDJoCo problem, and in particular its solution, have not yet been addressed in the literature (excluding our
recent conference papers \cite{cheng2015extension,cheng2016extension}, in which the extended SeDJoCo was first formulated, and a partial solution, which ignores some of the target matrices, was proposed). In
what follows we propose two comprehensive, general solution approaches, both based on extensions of existing iterative solutions of standard SeDJoCo. The first is an extension of
the Iterative Relaxation (IR) proposed by D{\'e}gerine and Za{\"\i}di \cite{degerine2004separation}. The second is based on Newton's method.

Note that both algorithms rely on some initial guess. A plausible option for obtaining an initial guess would be to initialize each $\B^{(m)}$ to the solution of the respective
standard SeDJoCo problem associated with the $m$-th set (thereby ignoring the information in the inter-set dependence). \vspace*{-0.3cm}
\subsection{Solution by Iterative Relaxations}
Recall the second formulation of the extended SeDJoCo problem \eqref{extended_sedjoco_alternative}, and notice that it can be written as
\begin{equation}\label{extended_sedjoco_alternative_orthogonal}
	\begin{aligned}
		\sum_{\ell=1}^{M}{\ub_{k_1}^{(m)\tps}}{\Q_{k_2}^{(m,\ell)}}\ub_{k_2}^{(\ell)}&={\ub_{k_1}^{(m)\tps}}\sum_{\ell=1}^{M}{\Q_{k_2}^{(m,\ell)}}\ub_{k_2}^{(\ell)}\\
		&={\ub_{k_1}^{(m)\tps}}\ubeta_{k_2}^{(m)}=\delta_{k_1k_2},\\
		&\forall k_1,k_2 \in \{1,\ldots,K\}, \forall m \in \{1,\ldots,M\}.
	\end{aligned}
\end{equation}
This means that the vector $\ub_{k}^{(m)}$ is orthogonal to all vectors $\left\{\ubeta_{k'}^{(m)}\right\}_{k'\neq k}$ and its inner product with $\ubeta_{k}^{(m)}$ equals 1.
Therefore, if we assume that all vectors $\left\{\ub_{k'}^{(m)}\right\}_{k' \neq k}$ are fixed, we can update the vector ${\ub_{k}^{(m)}}$ by a somewhat-similar Gram-Schmidt
procedure of subtracting its projection on the $K-1$ subspace spanned by these vectors, followed by a ``normalization" of the remaining residual. More precisely, the updating rule
is as follows:
\begin{equation}\label{update_rule_IR}
	\begin{aligned}
		& \ub_{k}^{(m)'} \leftarrow \ub_{k}^{(m)} - {\mBeta_{(k)}}^{(m)\tps}\left({\mBeta_{(k)}}^{(m)}{\mBeta_{(k)}}^{(m)\tps}\right)^{-1}{\mBeta_{(k)}}^{(m)}\ub_{k}^{(m)} \\
		& \ub_{k}^{(m)} \leftarrow \ub_{k}^{(m)'}/\sqrt{|{\ub_{k}^{(m)'\tps}}\ubeta_{k}^{(m)}|},
	\end{aligned}
\end{equation}
where the rows of the matrix ${\mBeta_{(k)}}^{(m)} \in \Rset^{K-1 \times K}$ are the $K-1$ vectors $\left\{\ubeta_{k'}^{(m)}\right\}_{k'\neq k}$. The update rule is repeated
iteratively until convergence, running through all $m \in \{1,\ldots,M\}$ and $k \in \{1,\ldots,K\}$.

\subsection{Solution by Newton's Method}
Let us define the gradient matrix of $\mathcal{L}(\bB)$ as $\G$, such that
\begin{equation} \label{grad_mat_of_likelihood}
	\begin{split}
		\frac{\partial \mathcal{L}(\bB)}{\partial \tB} \triangleq \G &= \left[\frac{\partial \mathcal{L}(\bB)}{\partial \B^{(1)}}\;\cdots\;\frac{\partial \mathcal{L}(\bB)}{\partial \B^{(M)}}\right]\\
		&\triangleq \left[\G^{(1)}\;\cdots\G^{(M)}\right] \in \Rset^{K \times KM},
	\end{split}
\end{equation}
such that $\G^{(\ell)} \in \Rset^{K \times K}, \forall \ell \in \{1,\ldots,M\}$. Differentiating further w.r.t. $\B^{(n)}$, we get the block-matrices of the Hessian
\begin{equation} \label{Hes_mat_of_likelihood}
	\H^{(n,\ell)} \triangleq \frac{\partial \mathcal{L}(\bB)}{\partial \B^{(n)} \partial \B^{(\ell)}} \in \Rset^{K^2 \times K^2}, \; \; \; \forall n,\ell \in \{1,\ldots,M\},
\end{equation}
where
\begin{equation} \label{Hes_permuted_def}
	\H \triangleq \begin{bmatrix}
		\H^{(1,1)} & \cdots & \H^{(1,M)}\\
		\vdots & \ddots & \vdots\\
		\H^{(M,1)} & \cdots & \H^{(M,M)}\end{bmatrix} \in \Rset^{K^2M \times K^2M}.
\end{equation}
More explicitly, we have
\begin{equation} \label{Hes_mat_of_likelihood}
	\begin{aligned}
		\H^{(n,\ell)} &= \frac{\partial}{\partial \B^{(n)}}\G^{(\ell)} \\
		&= \frac{\partial}{\partial \B^{(n)}}\left({\A^{(\ell)\tps}}-\sum_{k=1}^{K}\sum_{m=1}^{M}\E_{kk}\B^{(m)}\Q_{k}^{(m,\ell)}\right) \\
		&= -\left(\I_K\otimes{\A^{(\ell)\tps}}\right)\frac{\partial{\B^{(\ell)\tps}}}{\partial\B^{(n)}}\left(\I_K\otimes{\A^{(\ell)\tps}}\right)\\
		&\;\quad-\sum_{k=1}^{K}\sum_{m=1}^{M}\left(\I_K\otimes\E_{kk}\right)\frac{\partial\B^{(m)}}{\partial\B^{(n)}}\left(\I_K\otimes\Q_{k}^{(m,\ell)}\right) \\
		&= -\left(\I_K\otimes{\A^{(\ell)\tps}}\right)\left(\delta_{n\ell}\tE^{K}\right)\left(\I_K\otimes{\A^{(\ell)\tps}}\right)\\
		&\;\quad-\sum_{k=1}^{K}\sum_{m=1}^{M}\left(\I_K\otimes\E_{kk}\right)\left(\delta_{nm}\E^{K}\right)\left(\I_K\otimes\Q_{k}^{(m,\ell)}\right) \\
		&= -\delta_{n\ell}\left(\I_K\otimes{\A^{(\ell)\tps}}\right)\tE^{K}\left(\I_K\otimes{\A^{(\ell)\tps}}\right)\\
		&\;\quad-\sum_{k=1}^{K}\left(\I_K\otimes\E_{kk}\right)\E^{K}\left(\I_K\otimes\Q_{k}^{(n,\ell)}\right),
	\end{aligned}
\end{equation}
where we have used the property $\partial\A = -\A\partial\B\A$ \cite{petersen2008matrix}, and the notations
\begin{equation} \label{der_B_of_B}
	\frac{\partial\B^{(\ell)}}{\partial\B^{(n)}} = \delta_{n\ell} \begin{bmatrix}
		\E_{11} & \cdots & \E_{1M}\\
		\vdots & \ddots & \vdots\\
		\E_{M1} & \cdots & \E_{MM}\end{bmatrix} \triangleq \delta_{n\ell}\E^{K} \in \Rset^{K^2 \times K^2},
\end{equation}
and
\begin{equation} \label{der_B_transp_of_B}
	\frac{\partial{\B^{(m)\tps}}}{\partial\B^{(n)}} = \delta_{nm} \begin{bmatrix}
		\E_{11} & \cdots & \E_{M1}\\
		\vdots & \ddots & \vdots\\
		\E_{1M} & \cdots & \E_{MM}\end{bmatrix} \triangleq \delta_{nm}{\tE}^{K} \in \Rset^{K^2 \times K^2}.
\end{equation}
If we denote the columns of the matrix $\tB$ by $\tbeta_j$ for $1\le j \le KM$ and define the indexing function $j_{\left[k,m\right]}=\left(m-1\right)\cdot K +k$ for $1\le k \le
K$ and $1\le m \le M$, we have that
\begin{equation}
	\tbeta_{j_{[k,m]}} = \tbeta_k^{(m)},
\end{equation}
where $\tbeta_k^{(m)}$ denotes the $k$-th \textit{column} of the matrix $\B^{(m)}$. Consequently, we conclude that the vectorized gradient of $\text{vec}\left(\tB\right)$ is
$\text{vec}\left(\G\right)$, for which $\G$ is defined exactly as in \eqref{grad_mat_of_likelihood}. However, the Hessian matrix of the vector $\text{vec}\left(\tB\right)$ is
given by $\tH$, which is a permuted version of $\H$. More particularly, the elements of $\tH$ are given by 
\begin{equation}\label{hessian_elements}
	\begin{split}
		\tH_{\left(\text{ind}[p,q,i],\text{ind}[m,n,j]\right)} &= \H_{\left(\text{ind}[p,m,i],\text{ind}[q,n,j]\right)}\\
		&= \frac{\partial \mathcal{L}(\bB)}{\partial\B_{\left(p,q\right)}^{(i)}\partial\B_{\left(m,n\right)}^{(j)}},
	\end{split}
\end{equation}
and the subscript $_{\left(p,q\right)}$ denotes the $(p,q)$-th element of a matrix. Finally, the elements of the matrices $\B^{(1)},\cdots,\B^{(M)}$ are computed iteratively
according to Newton's update rule
\begin{equation}
	\text{vec}\left(\tB\right)^{\left[n+1\right]} = \text{vec}\left(\tB\right)^{\left[n\right]} + \Delta^{\left[n\right]},
\end{equation}
where
\begin{equation}
	\Delta^{\left[n\right]} = -\tH^{-1}\cdot \text{vec}\left(\G\right) \in \Rset^{K^2M}
\end{equation}
evaluated at $\text{vec}\left(\tB\right)^{\left[n\right]}$, and the superscript $^{\left[n\right]}$ indicates the $n$-th iteration.

\section{Simulation Results}
\label{sec_Simulation} We present simulation results of $3$ different experiments. First, we demonstrate the proposed algorithms' convergence behavior for a generic extended
SeDJoCo problem with random target-matrices. We then proceed to demonstrate the performance of the extended SeDJoCo solution as the ML estimate in the context of Gaussian JBSS
(IVA) in terms of common separation measures.
\subsection{Convergence Behavior}
In our first experiment we assess the convergence behavior of the two solutions proposed in Section \ref{sec_Solution} for a generic extended SeDJoCo problem. In this experiment
the target-matrices are generated as follows: First, we generate a set of $K$ PD matrices by
\begin{equation} \label{def_target_matrices_simul}
	\Omeg_k =\U_k\U_k^{\tps} \in \Rset^{KM \times KM}, \;\;\; \forall k \in \{1,\ldots,K\},
\end{equation}
where the elements of $\{\U_k\}_{k=1}^{K}$ are drawn independently from the standard Gaussian distribution. Then we take the $K \times K$ $KM^2$ blocks of all matrices
$\{\Omeg_k\}_{k=1}^{K}$, as defined in \eqref{augmented_taget}, to be the set of target-matrices. This way a solution is guaranteed to exist (according to \addra{Theorem} \ref{theorem1}).
We initialize the solution to be the set of identity matrices, i.e., $\{\B^{(m)}=\I_K\}_{m=1}^{M}$. Our measure of convergence is (cf. \eqref{result_e_sedjoco}) the logarithm of
\begin{equation} \label{Convergence_evaluation_metrix}
	\mathcal{E} \triangleq \left\|\sum_{m=1}^{M}\sum_{k=1}^K\sum_{\ell=1}^M\B^{(m)}\Q_k^{(m,\ell)}{\B^{(\ell)\tps}}\E_{kk}-\I_K \right\|_{\text{F}},
\end{equation}
which is the Frobenius norm of the residual-error matrix.

\begin{figure}[]
	\centering
	\begin{subfigure}[b]{0.23\textwidth}
		\includegraphics[width=\textwidth]{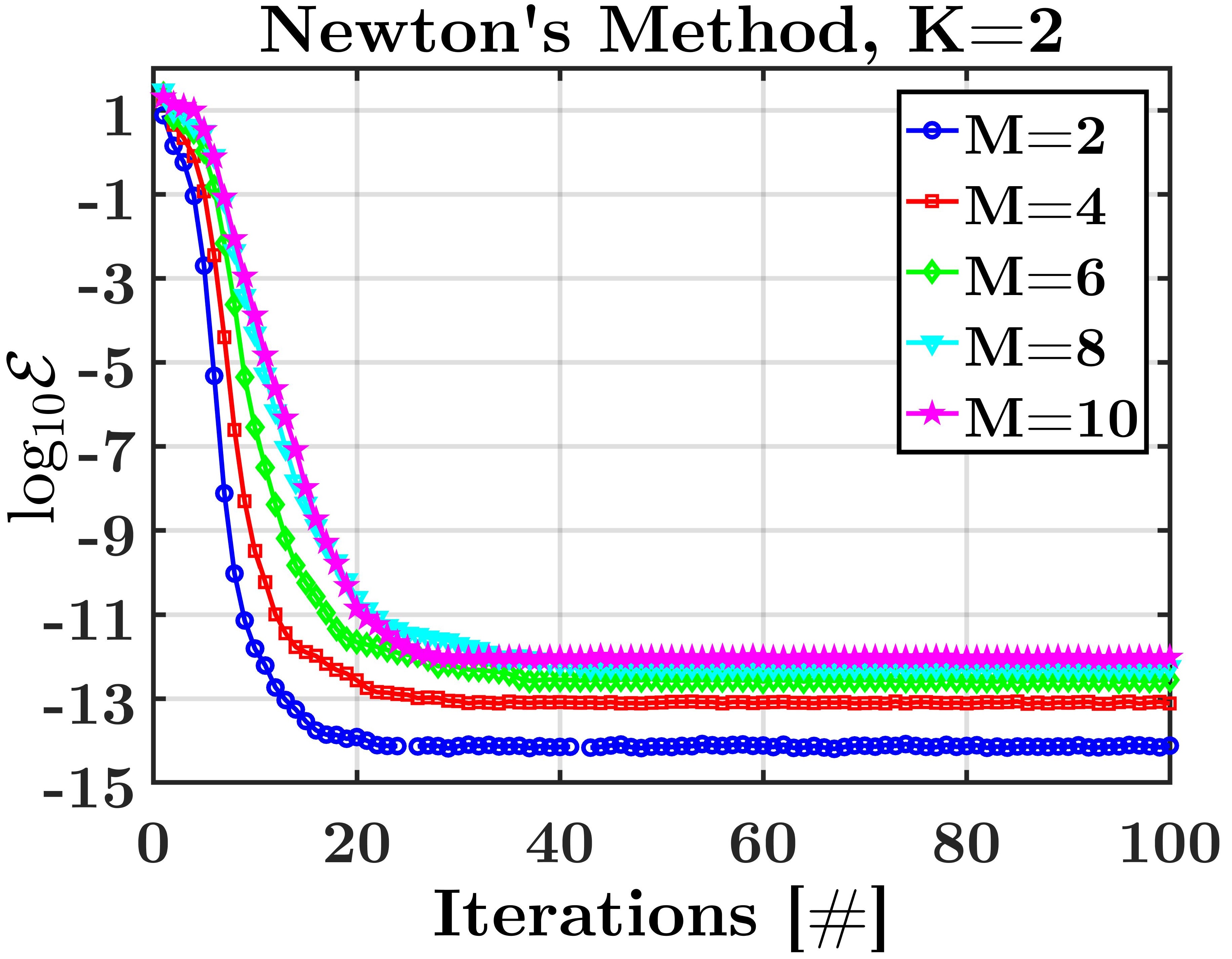}
		\caption{}
		\label{fig:Newton's_method_K_2}
	\end{subfigure}
	~
	\begin{subfigure}[b]{0.23\textwidth}
		\includegraphics[width=\textwidth]{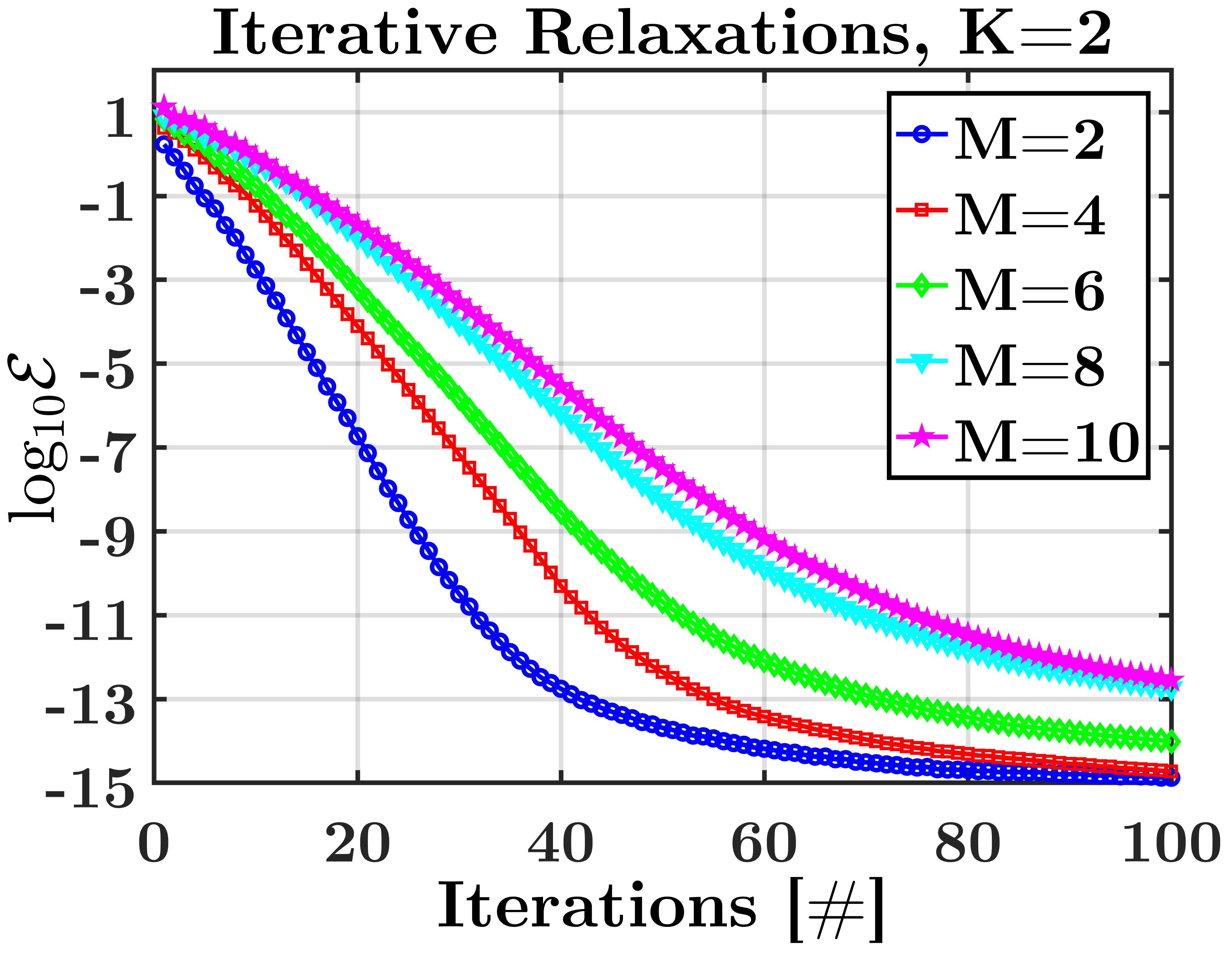}
		\caption{}
		\label{fig:IR_K_2}
	\end{subfigure}
	\begin{subfigure}[b]{0.23\textwidth}
		\includegraphics[width=\textwidth]{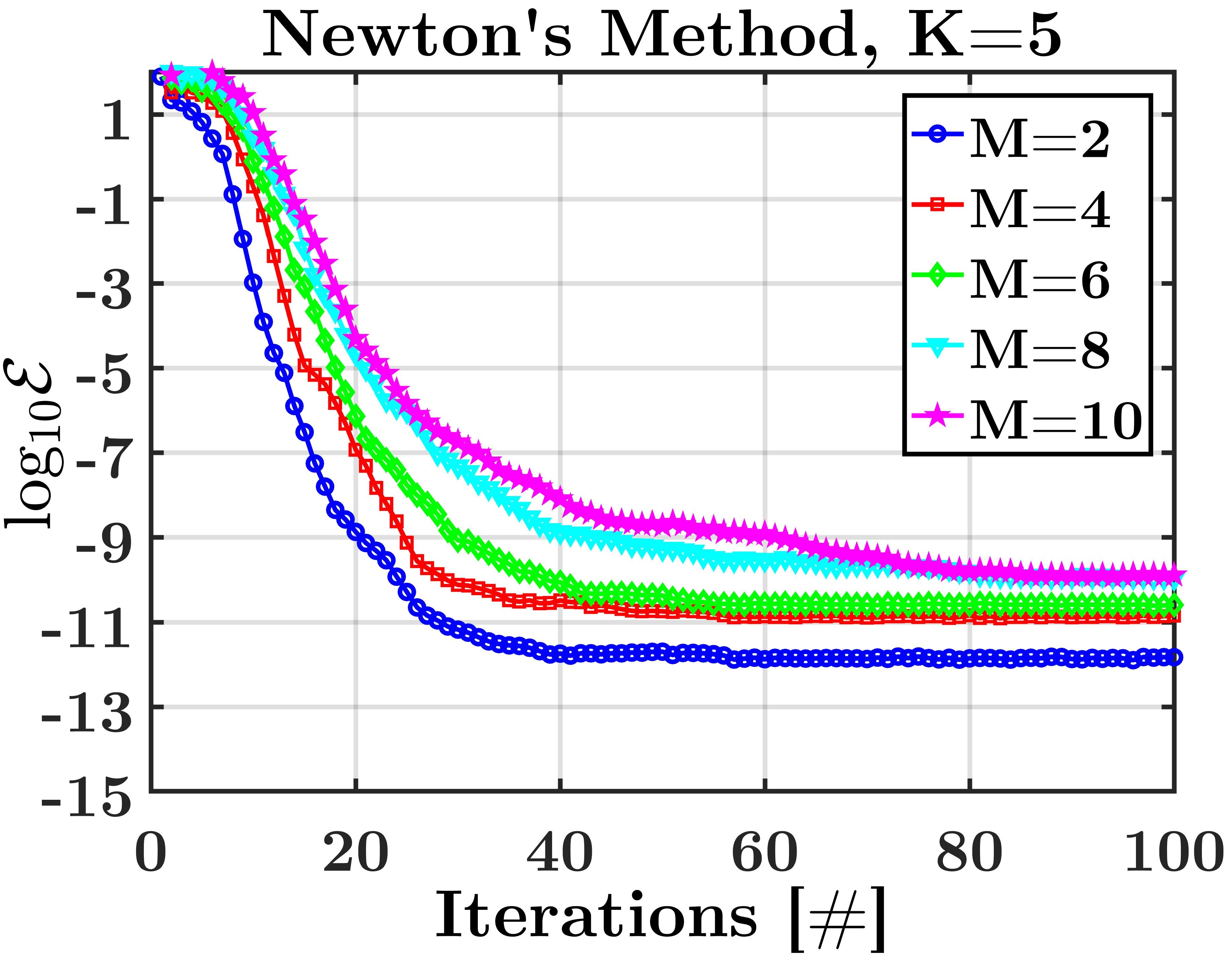}
		\caption{}
		\label{fig:Newton's_method_K_5}
	\end{subfigure}
	~
	\begin{subfigure}[b]{0.23\textwidth}
		\includegraphics[width=\textwidth]{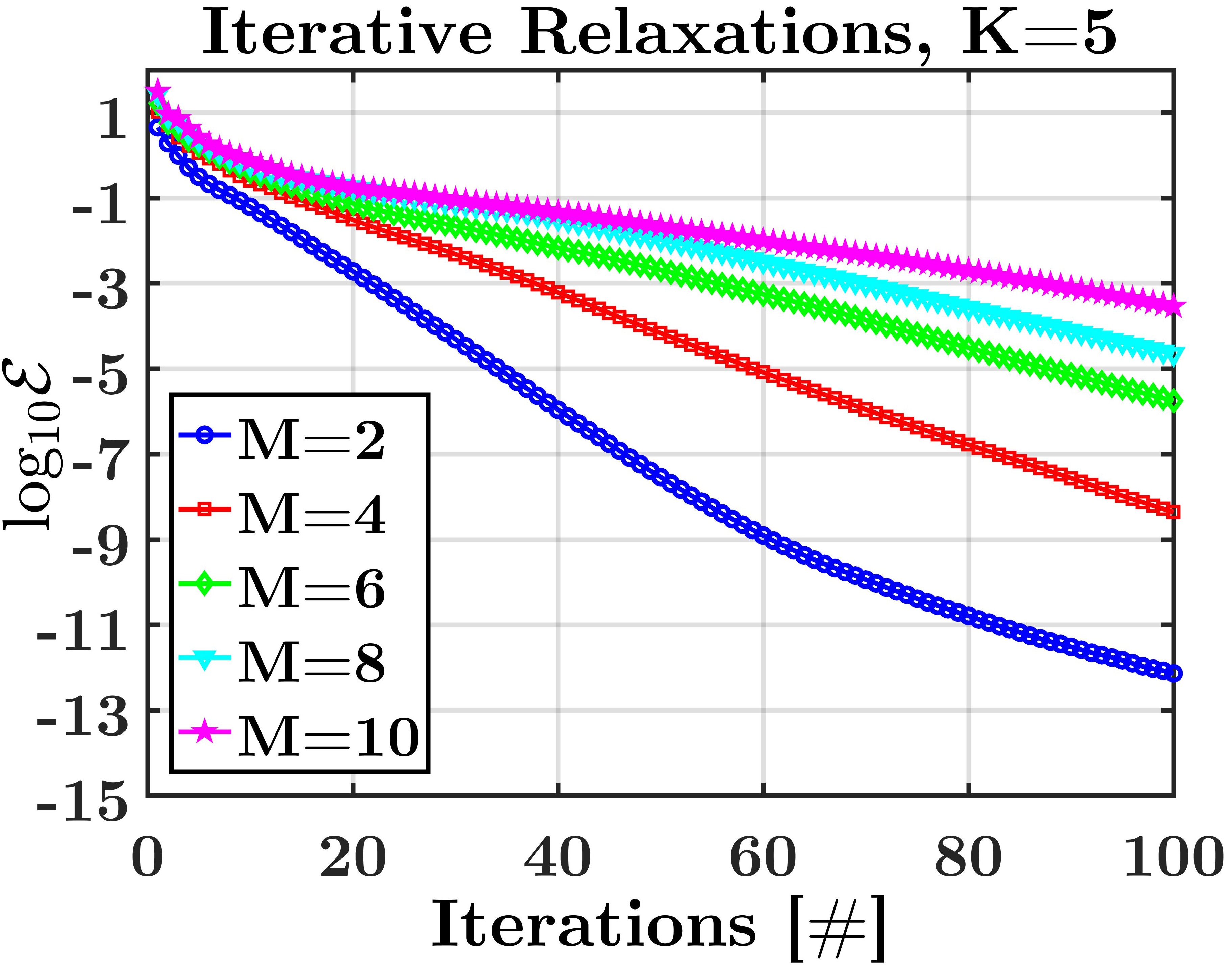}
		\caption{}
		\label{fig:IR_K_5}
	\end{subfigure}
	\caption{Convergence patterns of the proposed algorithm - logarithm of the residual root-mean-squares error vs. the number of iterations for different values of $M$ with a fixed value of $K$. (a) Newton's method, $K=2$ (b) IR, $K=2$ (c) Newton's method, $K=5$ (d) IR, $K=5$.}
	\label{fig:convergence_behavior}
\end{figure}

Fig. \ref{fig:convergence_behavior} shows the convergence patterns of the proposed algorithms for different values of $M$ with $K$ fixed.  The results in Figs.
\ref{fig:Newton's_method_K_2}-\ref{fig:IR_K_2} and \ref{fig:Newton's_method_K_5}-\ref{fig:IR_K_5} were obtained by averaging $100$ independent identical trials for $K=2$ and
$K=5$, respectively. It is evident that both algorithms converge to the solution; the Newton's algorithm converges much faster, typically within tens of iteration, whereas the IR
algorithm converges significantly more slowly as $K$ or $M$ or both increase. For this example, with $K=5$, $M=4$, and when convergence is defined at $\mathcal{E}=10^{-10}$, the
average run time for convergence was $0.5457$ seconds with Newton's method, and $1.7762$ seconds with the IR algorithm, so that the solution by Newton's method was computationally
more efficient in this case. Note, however, that the fast (quadratic) convergence of Newton's algorithm comes at the cost of a computational complexity increase per iteration; if we define a
full update iteration as an update of all the elements of $\tB$, the iterative relaxations algorithm requires $\mathcal{O}(MK^4)$ operations for a full update iteration (due to
$MK$ vector updates of $\mathcal{O}(K^3)$), whereas Newton's algorithm requires $\mathcal{O}(M^3K^6)$ (due to inversion of the Hessian). In addition, and as expected, it can be
seen that as $K$ or $M$ (or both) increase, the number of iterations increases as well
(for both algorithms). 
\subsection{JBSS (IVA) of Gaussian Sources}
In this part we focus on the application of the extended SeDJoCo as the ML solution for \addrb{semi-blind} Gaussian JBSS \addrb{(where the sources' covariance matrices are assumed to be known in advance)}. Results are based on averaging 1000 independent trials, where the mixing matrices elements are redrawn from a standard Gaussian distribution in each trial. We solve the extended SeDJoCo problem with the solution by Newton's method due to its (empirical) faster convergence and better resilience to initialization in convergence to the ML solution.
\begin{figure}[]
	\centering
	\begin{subfigure}[b]{0.4\textwidth}
		\includegraphics[width=\textwidth]{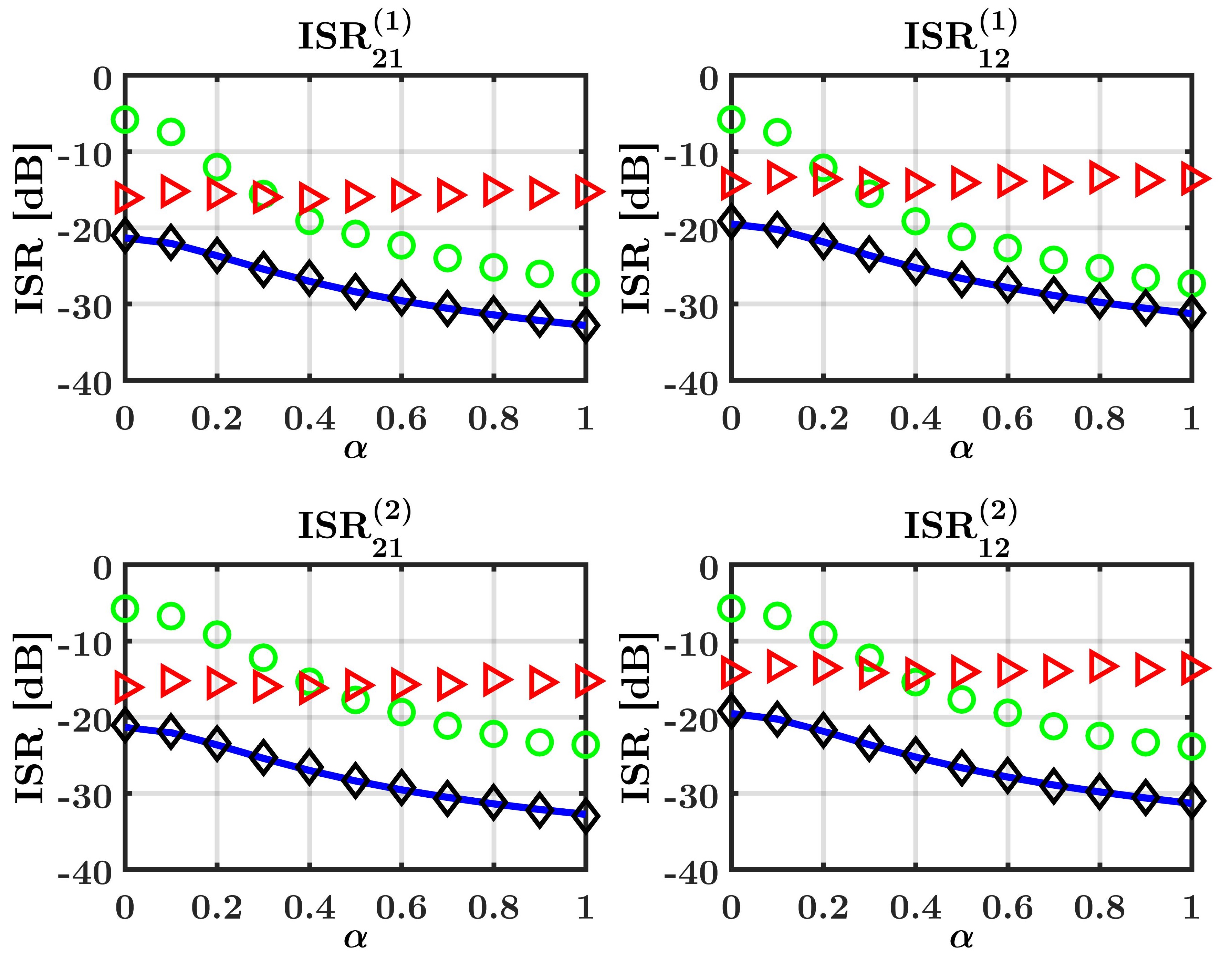}
		\caption{}
		\label{fig:ISR_elements_vs_alpha_non_stationary}
	\end{subfigure}
	\quad
	\begin{subfigure}[b]{0.4\textwidth}
		\includegraphics[width=\textwidth]{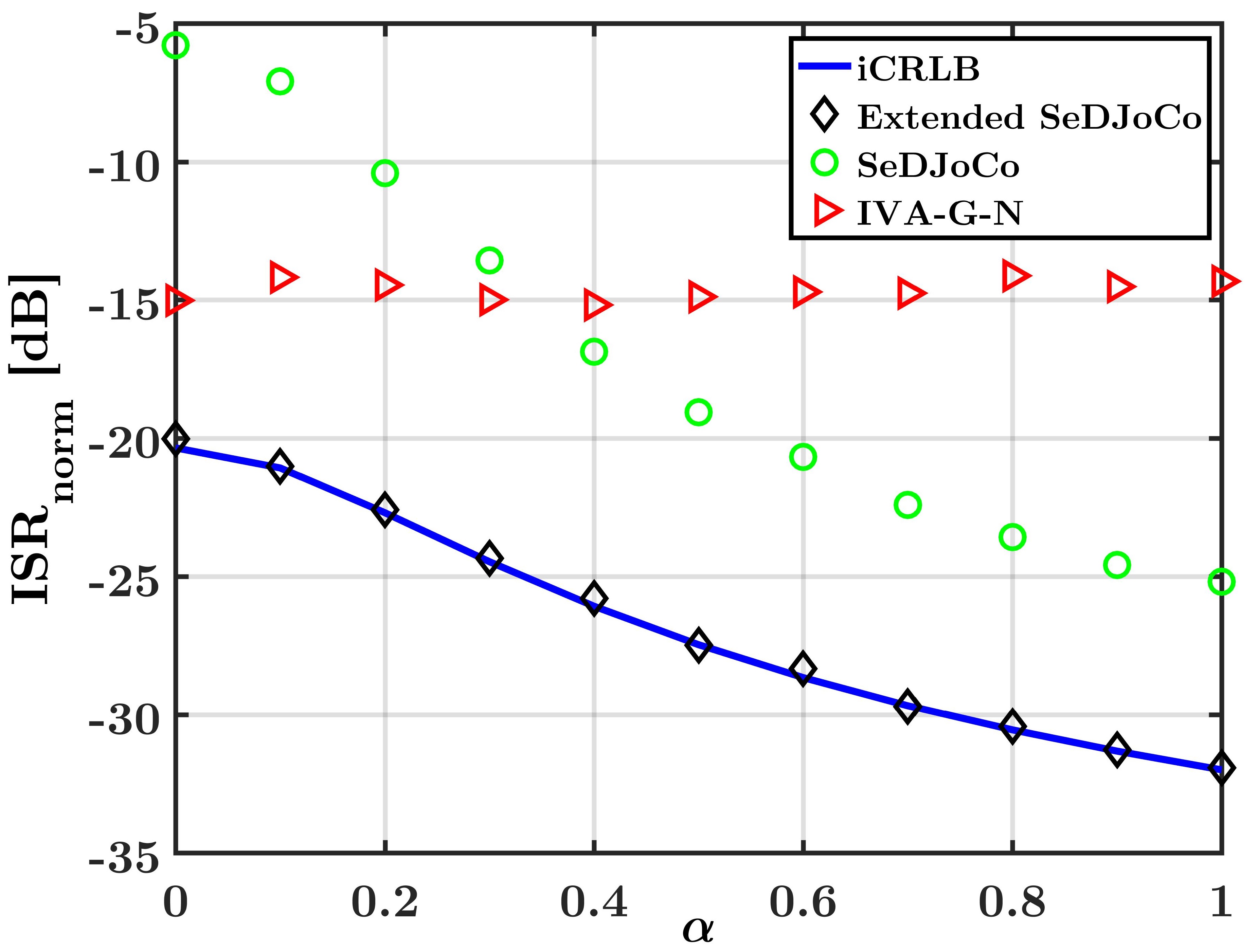}
		\caption{}
		\label{fig:ISR_norm_vs_alpha_non_stationary}
	\end{subfigure}
	\caption{(a) Empirical ISR vs. $\alpha$ (b) Empirical $\text{ISR}_{\text{norm}}$ vs. $\alpha$. It can be seen that the extended SeDJoCo solution is superior to both other solution. Furthermore, it attains the iCRLB. The legend in (b) is valid for (a) as well.}
	\label{fig:ISR_vs_alpha_non_stationary}
\end{figure}
\begin{figure}[]
	\centering
	\begin{subfigure}[b]{0.43\textwidth}
		\includegraphics[width=\textwidth]{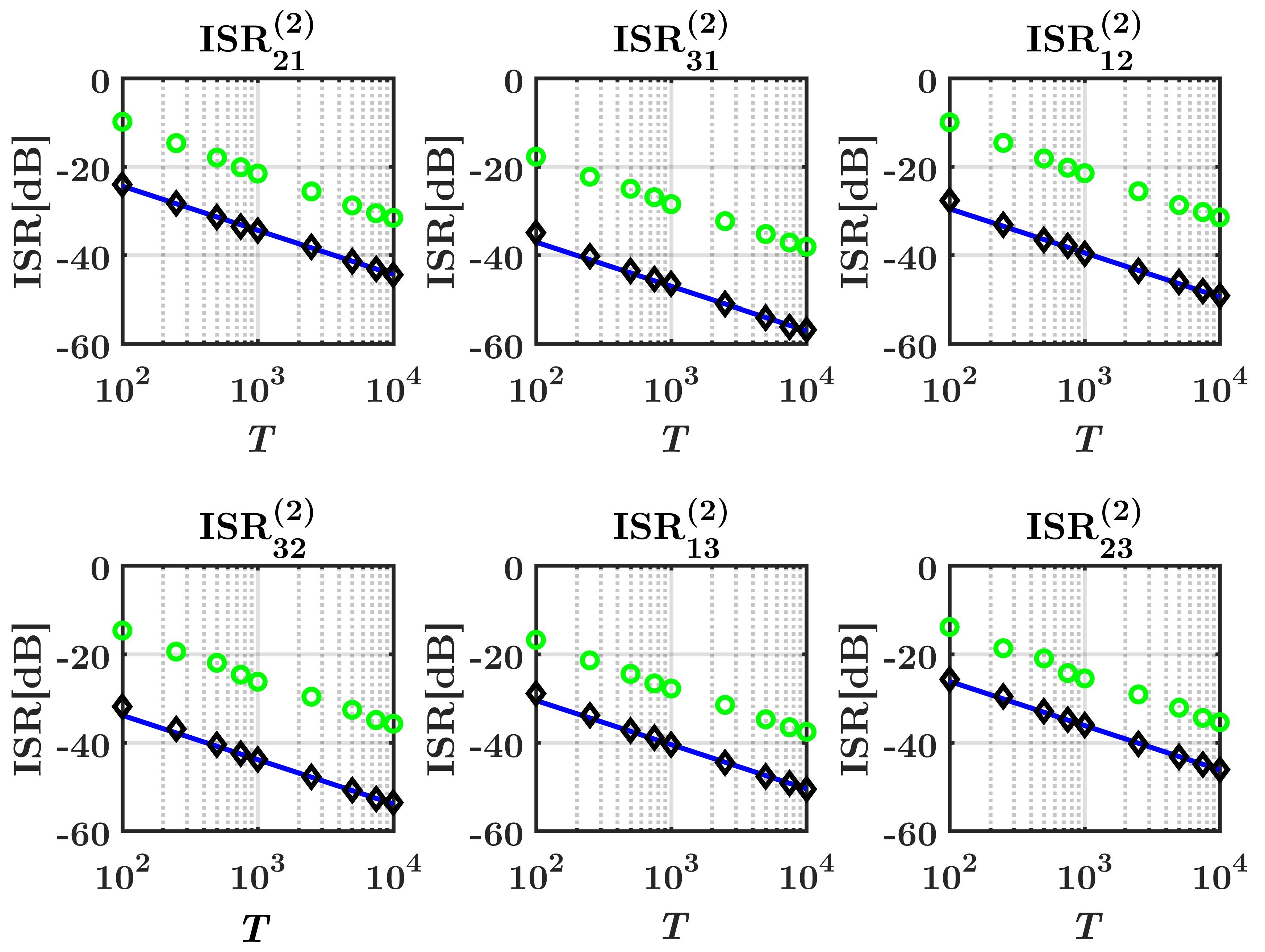}
		\caption{}
		\label{fig:ISR_vs_T_stationary}
	\end{subfigure}
	~
	\begin{subfigure}[b]{0.48\textwidth}
		\includegraphics[width=\textwidth]{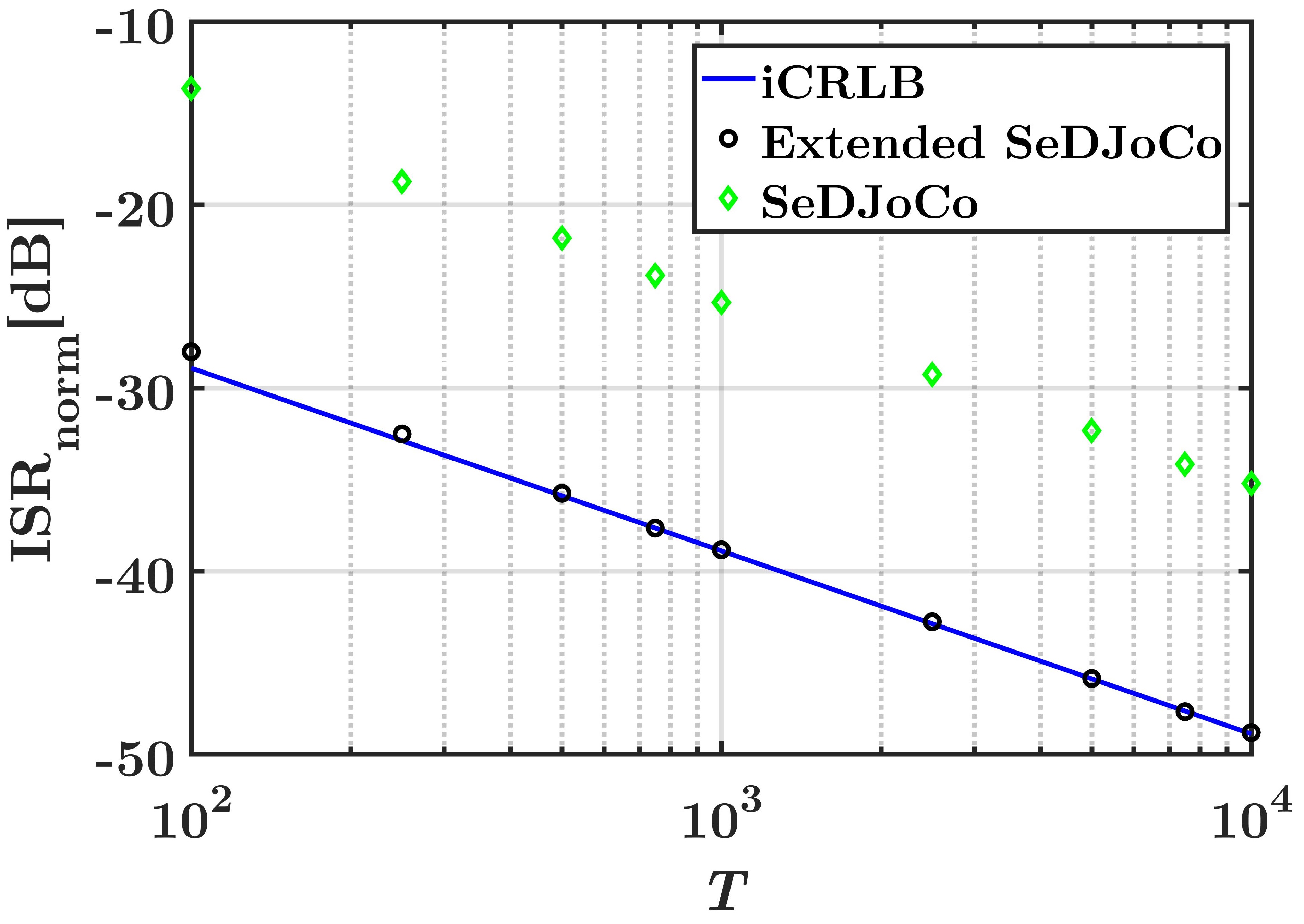}
		\caption{}
		\label{fig:ISR_norm_vs_T_stationary}
	\end{subfigure}
	\caption{(a) Empirical ISR vs. $T$ for the second dataset (b) Empirical $\text{ISR}_{\text{norm}}$ vs. $T$. The empirical results validate both the ML solution's optimality and our analytical expression of the iCRLB.}
	\label{fig:ISR_norm_vs_alpha_non_stationary}
	\vspace{-0.3cm}
\end{figure}
\begin{table}[]
	\centering
	\begin{tabular}{ | c | c | c | c | c |}
		\hline
		Parameter $\setminus$ $(k,m)$ & $(1,1)$ & $(1,2)$ & $(2,1)$ & $(2,2)$ \\ \hline
		$\phi_0^{(k,m)}$ & $\pi$ & $\frac{5\pi}{3}$ & $\frac{\pi}{3}$ & $\pi$ \\ \hline
		$N_{k,m}$ & $50$ & $350$ & $200$ & $500$ \\ \hline
		$\sigma_k^{(m)}$ & $2$ & $3\frac{1}{3}$ & $2\frac{2}{3}$ & $4$ \\ \hline
	\end{tabular}
	\caption{Fixed parameter values for experiment 2.}
	\label{table:FixedValuesSimul1}
\end{table}
First, we consider the simple case of two datasets ($M=2$), each with two Gaussian sources ($K=2$). The $k$-th source of the $m$-th dataset, $s_k^{(m)}[n]$, is generated as
\begin{multline} \label{source_definition_simul1}
	s_k^{(m)}[n] = \left(\sigma_k^{(m)}+\alpha \cdot \cos\left(\phi_{k,m}[n]\right)\right)w_k[n] + v_{k}^{(m)}[n],\\
	\forall k,m \in \{1,2\},
\end{multline}
where
\begin{equation} \label{source_definition_simul1_phi}
	\phi_{k,m}[n] = \frac{2\pi}{N_{k,m}}\cdot n + \phi_0^{(k,m)},
\end{equation}
$\alpha$ is a real parameter, $\left\{\phi_0^{(k,m)},N_{k,m},\sigma_k^{(m)}\right\}_{k,m=1}^2$ are fixed (see Table \ref{table:FixedValuesSimul1}), and
$\left\{w_k[n],v_k^{(m)}[n]\right\}_{k,m=1}^2$ are all mutually independent white standard Gaussian processes. Clearly, the $k$-th sources are correlated between sets, and when
$\alpha$ is non-zero the sources are non-stationary. In the limit case where $\alpha=0$, each source is a white Gaussian process, therefore separation cannot be attained by ICA
alone (i.e., when ignoring the inter-datasets correlations). We compare the performance of three different solutions: the extended SeDJoCo solution (by the proposed algorithm)
which yields the ML estimates w.r.t. the IVA problem altogether, the SeDJoCo solution of each dataset separately, which yields the ML estimates w.r.t. the two ICA problems
separately, and Anderson {\it et al.}'s \cite{anderson2012joint} Newton updates for Gaussian IVA (IVA-G-N), which is intended for separation of independent identically distributed
(i.i.d.) Gaussian sources. In this experiment we assume that the covariance matrices of the sources are known, i.e., a semi-blind scenario, and we demonstrate how the ML solution
can exploit this information in contrast to other solutions, e.g., the IVA-G-N, which cannot. The observation length was set to $T=1000$. Fig.
\ref{fig:ISR_vs_alpha_non_stationary} shows all the empirical ISR elements, as well as the empirical total normalized ISR,
\begin{equation} \label{JISR_defenition}
	\text{ISR}_{\text{norm}} \triangleq \frac{1}{MK(K-1)}\sum_{m=1}^{M}{\sum_{\substack{i,j=1 \\ i \neq j}}^{K}{\text{ISR}_{ij}^{(m)}}},
\end{equation}
vs. $\alpha \in [0,1]$, comparing also to the iCRLB derived in Section \ref{sec_iCRLB}. A good fit between the theoretical prediction and the empirical results is evident; when
$\alpha=0$ the SeDJoCo solutions (semi-blind ICA ML approach) collapse whereas the IVA approaches give good separation in terms of the ISR. The IVA-G-N performs properly since the
sources are indeed i.i.d., and the extended SeDJoCo solution (semi-blind IVA ML approach) achieves the iCRLB. As $\alpha$ increases, the sources become ``more" non-stationary.
Accordingly, IVA-G-N's performance becomes slightly worse (due to the model mis-match) while SeDJoCo keeps improving (to the point where the cross-correlations between datasets
are significantly less informative - compared to the temporal correlation of each source within each dataset). Extended SeDJoCo keeps attaining the iCRLB for all $\alpha$. We
stress that in this semi-blind scenario the ML solutions have an ``unfair" advantage over IVA-G-N, which, unlike the ML solutions, cannot exploit the prior knowledge of the
sources' covariance matrices. Nevertheless, it is our purpose in this work to show how available prior information such as this can be exploited in an optimal manner.

Our last experiment deals with zero-lag-uncorrelated stationary sources. We consider the case where $M=K=3$. The $k$-th source of the $m$-th dataset is generated as
\begin{equation} \label{source_definition_simul2}
	s_k^{(m)}[n] = v_k^{(m)}[n-L\cdot(m-1)], \;\; \forall k,m \in \{1,2,3\},
\end{equation}
where
\begin{equation} \label{colored_noise_simul2}
	v_k^{(m)}[n] = \sum_{\ell=1}^{M}{w_k^{(\ell)}[n]\ast h_k^{(m,\ell)}[n] },
\end{equation}
$\left\{w_k^{(m)}[n]\right\}_{k,m=1}^3$ are all mutually independent white, standard Gaussian noise processes, $\left\{h_k^{(m_1,m_2)}[n]\right\}_{k,m_1,m_2=1}^3$ are Finite
Impulse Response (FIR) filters of length $L$ for which
\begin{equation} \label{FIR_energy}
	\sum_{n=0}^{L-1}\left|h_k^{(m_1,m_2)}[n]\right|^2 = \begin{cases}
		1, & m_1=m_2 \\
		\eta, & m_1\neq m_2 \\
	\end{cases}, \forall k,m_1,m_2 \in \{1,2,3\},
\end{equation}
so that $\eta$ is a parameter which controls the ``relative energy" contained in the cross-spectra between corresponding sources from different datasets, and $\ast$ denotes the
convolution operator. As can be seen from \eqref{colored_noise_simul2}, $h_k^{(m_1,m_2)}[n]$ is the FIR filter applied to the $m_2$-th white driving-noise in order to generate a
component of the $k$-th source in the $m_1$-th dataset. Clearly, the cross-spectrum between any pair $\left\{s_k^{(m_1)}[n],s_k^{(m_2)}[n]:m_1 \neq m_2\right\}$ is non-zero when
$\eta>0$. However, note that although all such pairs are correlated, their zero-lag correlations are obviously zero (due to the $L\cdot(m-1)$ delays).

The FIR filters were drawn from a standard Gaussian distribution with $L=5$ and $\eta=1$. We compare the performance of the extended SeDJoCo solution and the SeDJoCo solutions
only, since the IVA-G-N algorithm requires instantaneous (zero-lag) correlation between sets and therefore performs very poorly in this scenario\footnote{This was validated in
	simulations.} (because, in addition to being zero-lag uncorrelated between sets, all sources have the same variance within sets, so they cannot even be ICA-separated using
IVA-G-N, due to its inherent temporal i.i.d. model assumption). As can be seen from Fig. \ref{fig:ISR_vs_T_stationary}, which shows the ISR elements of the second dataset vs. the
observation length $T$, the SeDJoCo solution yields quite good separation results using only the spectral diversity. However, the extended SeDJoCo solution ``enjoys" not only the
spectral diversity within each dataset, but also the cross-spectral diversity between the corresponding sources from different datasets. In this example, the average gain in ISR
is about $15$[dB] compared with the SeDJoCo solution. Similar results were obtained for the first and third datasets as well. This is reflected in Fig.
\ref{fig:ISR_norm_vs_T_stationary} which shows the total normalized ISR.
\vspace{-0.6cm}
\section{Conclusion}
\label{sec_Conclusion} We presented the ``extended SeDJoCo" problem, which is instrumental in finding the ML estimate of the separation matrices in the context of semi-blind IVA
in a Gaussian model. This problem is also closely related to CBF in a multicast setting, and possibly to other applications. Thus, after deriving different formulations of this general problem, we outlined some of its generic properties, such as a condition for the existence of a solution and multiplicity of the solutions. We also derived two iterative solution
algorithms, offering a trade-off between the required number of iterations and the computational complexity per iteration.

In the particular context of semi-blind IVA, we also derived the iCRLB on the elements of the ISR matrices for the case of Gaussian sources with arbitrary (but known) temporal
auto-covariance matrices and cross-covariance matrices (between sources in different sets). We then demonstrated how this broader paradigm enables (via a solution of the extended
SeDJoCo equations) the asymptotically optimal ML separation (attaining the iCRLB) of general stationary or non-stationary sources. This ability provides a significant advantage
over existing IVA methods, which so far only considered the model of temporally-i.i.d. source-vector components, and moreover, could not exploit prior knowledge in a semi-blind
scenario. \delra{We note further, that when the covariance matrices are not known {\it a-priori}, but can be succinctly parameterized (e.g., in the case of stationary parametric auto-regressive / moving average sources), an iterative separation strategy may be used, in which these matrices are first estimated from the data following initial separation, and then a semi-blind framework is applied using the estimated matrices, with successive refinements – thereby approaching asymptotic optimality.}
\vspace{-0.6cm}
\section{Acknowledgment}
\label{sec_acknowledgment} The authors gratefully acknowledge the financial support by the German-Israeli Foundation (GIF), grant number I-1282-406.10/2014. The first author also wishes to thank the Yitzhak and Chaya Weinstein Research Institute for Signal Processing for a fellowship. In addition, the authors wish to thank T\"{u}lay Adal{\i} for providing a helpful MATLAB code of the IVA-G-N algorithm.
\vspace{-0.5cm}
\appendices
\section{Differentiation of the Likelihood Function} \label{appendix_a}
Using the following properties:
\begin{equation*}
\begin{aligned}
& \text{(a)} \ \frac{\partial \log |\text{det}\X|}{\partial \X} = \left(\X^{-1}\right)^{\tps}, \text{(b)} \ \frac{\partial \ua^{\tps}\X\ub}{\partial \X} = \ua\ub^{\tps},\\
& \text{(c)} \ \frac{\partial \ua^{\tps}\X^{\tps}\ub}{\partial \X} = \ub\ua^{\tps}, \text{(d)} \ \Q_{k}^{(m_1,m_2)}={\Q_{k}^{(m_2,m_1)\tps}},\\
& \text{(e)} \ \frac{\partial \ub^{\tps}\X^{\tps}\D\X\uc}{\partial \X} = \D^{\tps}\X\ub\uc^{\tps}+\D\X\uc\ub^{\tps}, \; \text{(f)} \ \E_{ij} \triangleq \ue_i\ue_j^{\tps},\\
\end{aligned}
\end{equation*}
we have that
\begin{equation}
\label{d_log_likelihood}
\begin{aligned}
&\frac{\partial \mathcal{L}(\bB)}{\partial \B^{(m)}} = \frac{\partial}{\partial \B^{(m)}}\left(\sum_{\ell=1}^{M}\log|\text{det}\B^{(\ell)}|\right.\\
&\left.-\frac{1}{2}\sum_{k=1}^{K}\sum_{\substack{m_1=1\\m_2=1}}^{M}\ue_k^{\tps}\B^{(m_1)}\Q_{k}^{(m_1,m_2)}{\B^{(m_2)\tps}}\ue_k+\gamma\right) \\
&=\sum_{\ell=1}^{M}\frac{\partial \log|\text{det}\B^{(\ell)}|}{\partial \B^{(m)}} \\
&-\frac{1}{2}\sum_{k=1}^{K}\sum_{\substack{m_1=1\\m_2=1}}^{M}\frac{\partial\ue_k^{\tps}\B^{(m_1)}\Q_{k}^{(m_1,m_2)}{\B^{(m_2)\tps}}\ue_k}{\partial \B^{(m)}} \\
&\underset{\text{(a)}}{=} {\A^{(m)\tps}}-\frac{1}{2}\sum_{k=1}^{K}\sum_{m_1=1}^{M}\sum_{\substack{ m_2=1 \\ m_2 \ne m_1}}^{M}\frac{\partial\ue_k^{\tps}\B^{(m_1)}\Q_{k}^{(m_1,m_2)}{\B^{(m_2)\tps}}\ue_k}{\partial \B^{(m)}} \\
&-\frac{1}{2}\sum_{k=1}^{K}\sum_{m_1=1}^{M}\frac{\partial}{\partial \B^{(m)}}\left(\ue_k^{\tps}\B^{(m_1)}\Q_{k}^{(m_1,m_1)}{\B^{(m_1)\tps}}\ue_k\right) \\
&\underset{\text{(b),(c),(d),(e)}}{=} {\A^{(m)\tps}}-\frac{1}{2}\sum_{k=1}^{K}2\sum_{m_1=1}^{M}\ue_k\ue_k^{\tps}\B^{(m_1)}\Q_k^{(m_1,m)}\\
&\underset{\text{(f)}}{=} {\A^{(m)\tps}}-\sum_{k=1}^{K}\sum_{m_1=1}^{M}\E_{kk}\B^{(m_1)}\Q_k^{(m_1,m)}.
\end{aligned}
\end{equation}
\vspace{-0.8cm}
\bibliography{Bibfile}
\bibliographystyle{unsrt}

\end{document}